\documentclass[10pt]{article}

\usepackage{amsmath,amsfonts,amssymb,amsthm,graphicx,stmaryrd,epigraph}
\usepackage{enumerate}
\usepackage{etoolbox}
\usepackage{color}
\usepackage[shortcuts]{extdash}
\usepackage{algorithm}
\usepackage[noend]{algpseudocode}  
\algrenewcommand\algorithmicthen{\relax}
\algdef{C}[IF]{IF}{ElsIf}[1]{\textbf{elif}\ #1}
\algrenewcommand\algorithmicdo{\relax}

\usepackage[colorlinks=true,citecolor=blue,pdfpagemode=UseNone,pdfstartview=FitH]{hyperref}

\allowdisplaybreaks[4]

\renewcommand{\d}{\,\mathrm{d}}
\newcommand{\dd}{\mathrm{d}}

\newcommand{\E}{\mathbb{E}}

\DeclareMathOperator{\inv}{in}
\DeclareMathOperator{\outv}{out}
\newcommand{\source}{\textnormal{source}}
\newcommand{\sink}{\textnormal{sink}}

\DeclareMathOperator{\ep}{ep}
\DeclareMathOperator{\mep}{mep}

\DeclareMathOperator{\KL}{\mathrm{KL}}

\newcommand{\LB}{\mathrm{LB}}
\newcommand{\ELB}{\mathrm{ELB}}
\newcommand{\UB}{\mathrm{UB}}
\newcommand{\UMM}{\mathrm{UMM}}

\newcommand{\Beta}{\mathrm{B}}

\newcommand{\peq}{\mathrel{+}=}

\newcommand{\Sum}{\textrm{Sum}}

\theoremstyle{plain}
\newtheorem{theorem}{Theorem}[section]
\newtheorem{corollary}[theorem]{Corollary}
\newtheorem{lemma}[theorem]{Lemma}
\newtheorem{proposition}[theorem]{Proposition}

\theoremstyle{definition}

\newtheorem{example}[theorem]{Example}

\theoremstyle{remark}
\newtheorem{remark}[theorem]{Remark}

\newcommand{\abstr}
  {The topic of this paper is testing exchangeability
  using e-values in the batch mode,
  with the Markov model as alternative.
  The null hypothesis of exchangeability
  is formalized as a Kolmogorov-type compression model,
  and the Bayes mixture of the Markov model
  w.r.\ to the uniform prior is taken as simple alternative hypothesis.
  Using e-values instead of p-values leads
  to a computationally efficient testing procedure.
  In the appendixes I explain connections with the algorithmic theory of randomness
  and with the traditional theory of testing statistical hypotheses.
  In the standard statistical terminology,
  this paper proposes a new permutation test.
  This test can also be interpreted as a poor man's version
  of Kolmogorov's deficiency of randomness.}

\begin{document}
\title{Testing exchangeability in the batch mode with e-values and Markov alternatives}
\author{Vladimir Vovk}

\maketitle
\begin{abstract}
  \smallskip
  \abstr

  The version of this paper at \url{http://alrw.net} (Working Paper 38)
  is updated most often.
\end{abstract}

\section{Introduction}

The usual approach to testing exchangeability in statistics
is based on using p-values, as in \cite[Sect.~7.2]{Lehmann:2006}.
In this paper we will use e-values instead
\cite{Vovk/Wang:2021,Grunwald/etal:arXiv1906},
which facilitates computations.
E-values have been used for testing exchangeability via conformal prediction
\cite[Part~III]{Vovk/etal:2022book} in the online protocol,
while in this paper we will use the standard batch protocol:
we are given the data sequence as one batch
rather than getting its elements sequentially one by one.

The null hypothesis of exchangeability will be defined
in Sect.~\ref{sec:testing} using the terminology of compression modelling
\cite[Chap.~11]{Vovk/etal:2022book}.
Compression modelling is an algorithm-free version of Kolmogorov's way of stochastic modelling:
cf.\ \cite{Vovk:2001Denmark}, \cite{Vovk/Shafer:2003}, \cite[Sect.~2]{Vyugin:arXiv1907},
and \cite[Sect.~11.6.1]{Vovk/etal:2022book}.
Kolmogorov's original version will be discussed in Appendix~\ref{app:ATR}.

In Sect.~\ref{sec:testing} we also define e-variables,
our tools for obtaining e-values in testing exchangeability
(or another null hypothesis).
We will derive our main e-variable as likelihood ratio
for a Markovian alternative hypothesis,
which we will introduce in Sect.~\ref{sec:algorithm}.
A simple optimality property of the likelihood ratios
is derived in Sect.~\ref{sec:performance}.

After defining our main alternative hypothesis in Sect.~\ref{sec:algorithm},
we derive an efficient algorithm for computing the corresponding e-variable.
The power of this e-variable is the topic of Sect.~\ref{sec:e-power}.
The algorithm's performance in view of the results of Sect.~\ref{sec:e-power}
is studied in Sect.~\ref{sec:experiments} using simulated data.
Section~\ref{sec:conclusion} concludes.

In Appendix~\ref{app:ATR} I describe Kolmogorov's original ideal picture
of algorithmic randomness.
In the following Appendix~\ref{app:universal} we will discuss possible ways
of making this picture more practical,
and in Appendix~\ref{app:changepoint} will go deeper
into another class of alternatives for testing exchangeability
(namely, into the changepoint alternatives).

In traditional statistics, the p-value version of the procedure of this paper
is often presented in terms of the Neyman structure;
see, e.g., \cite[Sect.~4.3]{Lehmann/Romano:2022}.
We discuss its counterpart for e-values in Appendix~\ref{app:Neyman}.

\section{Testing exchangeability}
\label{sec:testing}

We consider the simplest binary case,
and our \emph{observation space} is $\mathbf{Z}:=\{0,1\}$.
Fix an integer $N>1$, which we will refer to as the \emph{horizon}.
We are interested in binary \emph{data sequences}
$(z_1,\dots,z_N)\in\Omega:=\mathbf{Z}^N$.
A \emph{Kolmogorov compression model} (KCM)
is a \emph{summarising statistic} $t:\Omega\to\Sigma$,
where $\Sigma$ is a finite set (the \emph{summary space}),
together with the implicit statement
that given the \emph{summary} $t(z_1,\dots,z_N)$
(for which we do not make any stochastic assumptions)
the actual data sequence $(z_1,\dots,z_N)$ is generated
from the uniform probability measure.
Our \emph{null hypothesis} is the KCM,
which we call the \emph{Kolmogorov exchangeability model},
$t_E(z_1,\dots,z_N):=z_1+\dots+z_N$.

Let us say that a probability measure $P$ \emph{agrees}
with a summarising statistic $t$
if the data sequences with the same summary have the same $P$-probability.

\begin{lemma}\label{lem:mixtures}
  The exchangeable probability measures on $\Omega$
  are exactly the probability measures that agree
  with the Kolmogorov exchangeability model
  (the mixtures of the uniform probability measures on $t_E^{-1}(k)$).
\end{lemma}

The easy proof of Lemma~\ref{lem:mixtures} is omitted.
It shows that, in terms of standard statistical modelling,
we can define our null hypothesis as the set of all exchangeable probability measures
on $\Omega$.

An \emph{e-variable} w.r.\ to a probability measure
is a nonnegative function on $\Omega$ with expectation at most 1.
An \emph{exchangeability e-variable} is a function $E:\Omega\to[0,\infty)$
whose average over each $t_E^{-1}(k)$ is at most 1.
Alternatively, it is an exchangeability e-variable
w.r.\ to any exchangeable probability measure.

\begin{proposition}
  The two meanings of an exchangeability e-variable coincide.
\end{proposition}

\begin{proof}
  If the average of $E$ over each $t_E^{-1}(k)$ is at most 1,
  it will be an e-variable w.r.\ to each exchangeable probability measure
  by Lemma~\ref{lem:mixtures}.

  Now suppose $E$ is an e-variable w.r.\ to each exchangeable probability measure.
  Since the uniform probability measure on $t_E^{-1}(k)$
  is exchangeable,
  the average of $E$ over $t_E^{-1}(k)$ will be at most 1.
\end{proof}

All null hypotheses considered in this paper will be Kolmogorov compression models.
In the main part of the paper we will concentrate on the exchangeability model,
but in this and next section we will also give more general definitions.
An \emph{e-variable} w.r.\ to a KCM $t$ is a function $E:\Omega\to[0,\infty)$
such that the arithmetic mean of $E$ over $t^{-1}(\sigma)$ is at most 1
for any $\sigma\in t(\Omega)$.
\emph{E-values} are values taken by e-variables.

\subsection*{Disintegration of the alternative hypothesis}

Let us fix an \emph{alternative hypothesis} $Q$,
which is a probability measure on $\Omega$.
Our statistical procedures will depend on $Q$
only via the corresponding {batch compression model} (BCM).
A BCM is a pair $(t,P)$ such that $t:\Omega\to\Sigma$ is a summarising statistic
and $P:\Sigma\hookrightarrow\Omega$
(to use the notation of \cite[Sect.~A.4]{Vovk/etal:2022book})
is a Markov kernel
such that $P(\sigma)$ is concentrated on $t^{-1}(\sigma)$
for each $\sigma\in\Sigma$.
As before, we refer to $t(\omega)$ as the \emph{summary} of $\omega$.
Kolmogorov compression models are a special case
in which $P(\sigma)$ are the uniform probability measures on $t^{-1}(\sigma)$.

\begin{remark}
  Batch compression models are in fact standard
  and are often used without giving them any name, as in \cite{Lauritzen:1988}.
\end{remark}

With an alternative hypothesis $Q$ and a null hypothesis $t$
we associate the \emph{alternative Markov kernel}
\[
  Q_{\sigma}(\{\omega\})
  :=
  \frac{Q(\{\omega\})}{Q(t^{-1}(\sigma))},
  \quad
  \sigma\in\Sigma,
  \enspace
  \omega\in t^{-1}(\sigma).
\]
As compared with $Q$, the alternative Markov kernel
loses the information about $Q(\{\sigma\})$ for $\sigma\in\Sigma$.

\section{Frequentist performance of e\-/variables}
\label{sec:performance}

Suppose $Q$ (the alternative probability measure)
is the true data-generating distribution
and we keep generating data sequences $(z_1,\dots,z_N)\in\Omega$
from $Q$ in the IID fashion.
The following lemma allows us to define the efficiency of an e-variable
via its frequentist performance when we keep applying it repeatedly
to accumulate capital.
This is a special case of Kelly's criterion \cite{Kelly:1956}.

\begin{lemma}\label{lem:frequentist}
  Consider an e-variable $E$
  w.r.\ to a Kolmogorov compression model $t:\Omega\to\Sigma$.
  For any alternative probability measure $Q$ on $\Omega$,
  the limit\footnote{%
    In this paper, our notation for logarithms is $\ln$ (natural)
    and $\log$ (binary, in Appendix~\ref{app:ATR}).}
  \begin{equation}\label{eq:ep}
    \ep_Q(E)
    :=
    \lim_{I\to\infty}
    \frac1I
    \ln
    \prod_{i=1}^I
    E(z_1^i,\dots,z_N^i)
  \end{equation}
  where $(z_1^i,\dots,z_N^i)$ is the $i$th data sequence generated from $Q$ independently,
  exists $Q^{\infty}$-almost surely.
  Moreover, for all $E$ and $Q$,
  \begin{equation}\label{eq:e-power}
    \ep_Q(E)
    =
    \int \ln E \d Q.
  \end{equation}
\end{lemma}

The interpretation of \eqref{eq:ep} is that
our capital
$
  \prod_{i=1}^I
  E(z_1^i,\dots,z_N^i)
$
grows exponentially fast
(we will see later, in Lemma~\ref{lem:optimal},
that we can indeed expect it to grow rather than shrink
if we can guess a good $Q$),
and its rate of growth is given by the expression \eqref{eq:e-power},
which we will refer to as the \emph{e-power} of $E$ under the alternative $Q$.

\begin{proof}
  It suffices to rewrite \eqref{eq:ep} as
  \[
    \ep_Q(E)
    =
    \lim_{I\to\infty}
    \frac1I
    \sum_{i=1}^I
    \ln E(z_1^i,\dots,z_N^i)
  \]
  and apply Kolmogorov's law of large numbers
  to the IID random variables $\ln E(z_1^i,\dots,z_N^i)$
  with expectation $\int \ln E \d Q$
  (which exists and is finite since the sample space is assumed to be finite).
\end{proof}

To justify the expression \eqref{eq:e-power}
using frequentist considerations, we do not really need the IID picture,
as emphasized by Neyman \cite[Sect.~10]{Neyman:1977}.
When generating $z_1^i,\dots,z_N^i$ for different $i$,
we may test different Kolmogorov compression models $t=t_i$,
perhaps with different horizons $N=N_i$,
against different alternatives $Q=Q_i$.
The corresponding generalization of Lemma~\ref{lem:frequentist}
states that the long-term rate of growth of our capital
will be asymptotically close
to the arithmetic average of $\int \ln E_i \d Q_i$.
It will involve certain regularity conditions
needed for the applicability of the martingale strong law of large numbers
(e.g., in the form of \cite[Chap.~4]{Shafer/Vovk:2019},
which allows non-stochastic choice of $N_i$, $t_i$, and $Q_i$)
If the alternative hypothesis does not hold in all trials,
Lemma~\ref{lem:frequentist} is still applicable to the trials
where it does hold.

Now it is easy to find the optimal, in the sense of $\ep_Q$, e-variable;
it will be the ratio of the alternative Markov kernel to the null hypothesis.

\begin{lemma}\label{lem:optimal}
  The maximum of $\ep_Q$ is attained at
  \begin{equation}\label{eq:E}
    E(\omega)
    :=
    \left|
      t^{-1}(t(\omega))
    \right|
    Q_{t(\omega)}(\{\omega\}),
    \quad
    \omega\in\Omega.
  \end{equation}
  In this case,
  \begin{equation}\label{eq:super-expression}
    \mep(Q)
    :=
    \ep_Q(E)
    =
    \int
    \ln
    \left|
      t^{-1}(\sigma)
    \right|
    (t_*Q)(\dd\sigma)
    +
    H(t_*Q)
    -
    H(Q),
  \end{equation}
  where $t_*Q$ stands for the push-forward measure of $Q$ by $t$
  (the summarising statistic of the null hypothesis),
  and $H$ stands for the entropy.
\end{lemma}

We will call $\mep(Q)$ defined by \eqref{eq:super-expression}
the \emph{maximum e-power} of the alternative $Q$.
A sizeable $\mep(Q)$ for a plausible alternative $Q$ means that the testing problem is not hopeless
and has some potential.
The guarantee given by Lemma~\ref{lem:frequentist}, however,
is frequentist and not applicable
if testing is done only once,
in which case we also want the optimal e-variable \eqref{eq:E}
not to be too volatile.

\begin{proof}
  In this paper we let $U_A$
  stand for the uniform probability measure on a finite set $A$.
  The optimization $\int E \d Q\to\max$
  can be performed inside each block $t^{-1}(\sigma)$ separately.
  Using the nonnegativity of the Kullback--Leibler divergence,
  we have, for each $\sigma\in t(\Omega)$,
  \[
    \ep_{Q_{\sigma}}
    \left(
      \frac{Q_{\sigma}}{U_{t^{-1}(\sigma)}}
    \right)
    \ge
    \ep_{Q_{\sigma}}(E')
  \]
  for each e-variable $E'$ w.r.\ to $t$,
  which implies the first statement (about \eqref{eq:E}) of the lemma.
  The second statement \eqref{eq:super-expression} follows from
  \begin{align*}
    \ep_Q(E)
    &=
    \int
    \KL(Q_{\sigma} \mathbin{\|} U_{t^{-1}(\sigma)})
    (t_*Q)(\dd\sigma)\\
    &=
    \int
    \left(
      \ln
      \left|
        t^{-1}(\sigma)
      \right|
      -
      H(Q_{\sigma})
    \right)
    (t_*Q)(\dd\sigma)\\
    &=
    \int
    \ln
    \left|
      t^{-1}(\sigma)
    \right|
    (t_*Q)(\dd\sigma)
    +
    H(t_*Q)
    -
    H(Q),
  \end{align*}
  where $\KL$ stands for the Kullback--Leibler divergence.
\end{proof}

\section{An explicit algorithm for Markov alternatives}
\label{sec:algorithm}

Starting from this section we will consider a specific alternative hypothesis
obtained by mixing Markov probability measures.
The corresponding exchangeability e-variable
will be computable in linear time, $O(N)$.

First let us fix some terminology.
The \emph{exchangeability summary},
or \emph{exchangeability type},
of a data sequence $z_1,\dots,z_N$ is the numbers $(N_0,N_1)$
of 0s and 1s in it.
(It carries the same information as just the number of 1s,
but we prefer a symmetric definition despite some redundancy.)
By a ``substring'' we always mean a contiguous substring.
The \emph{Markov type} of $z_1,\dots,z_N$ is the sextuple
$(F,N_{00},N_{01},N_{10},N_{11},L)$,
where $N_{i,j}$  is the number of times $(i,j)$
occurs as substring in the sequence $z_1,\dots,z_N$
(with the comma often omitted),
and $F$ and $L$ are the first and last bits.

As our alternative hypothesis,
we will take the uniform mixture of the Markov probability measures,
defined as follows:
$\pi_{01}$ and $\pi_{10}$ are generated independently
from the uniform distribution $U_{[0,1]}$ on $[0,1]$;
the first bit is chosen as $1$ with probability $1/2$,
and after that each $0$ is followed by $1$ with probability $\pi_{01}$,
and each $1$ is followed by $0$ with probability $\pi_{10}$.
Let us compute the probability of a sequence of a Markov type
$(F,N_{00},\dots,N_{11},L)$ under this probability measure:
\begin{equation}\label{eq:alternative}
  \begin{aligned}
    \frac12
    \int
    &(1-\pi_{01})^{N_{00}}
    \pi_{01}^{N_{01}}
    \pi_{10}^{N_{10}}
    (1-\pi_{10})^{N_{11}}
    \d\pi_{01}
    \dd\pi_{10}\\
    &=
    \frac12
    \Beta(N_{00}+1,N_{01}+1)
    \Beta(N_{10}+1,N_{11}+1)\\
    &=
    \frac12
    \frac{
      \Gamma(N_{00}+1)
      \Gamma(N_{01}+1)
      \Gamma(N_{10}+1)
      \Gamma(N_{11}+1)
    }
    {
      \Gamma(N_{0*}+2)
      \Gamma(N_{1*}+2)
    }\\
    &=
    \frac12
    \frac{
      N_{00}!
      N_{01}!
      N_{10}!
      N_{11}!
    }
    {
      (N_{0*}+1)!
      (N_{1*}+1)!
    },
  \end{aligned}
\end{equation}
where $N_{i*}:=N_{i,0}+N_{i,1}$.
If $N_{1-F}=0$, this probability is $\frac{1}{2N}$
(which in fact agrees with the general expression \eqref{eq:alternative}).

For future use, set $\pi_{00}:=1-\pi_{01}$ and $\pi_{11}:=1-\pi_{10}$.

The expression \eqref{eq:alternative} gives us,
analogously to \cite[Chap.~9]{Vovk/etal:2022book}
(who follow \cite{Ramdas/etal:2022}),
the \emph{lower benchmark}
\begin{equation}\label{eq:LB}
  \LB
  :=
  \frac12
  \frac{
    N_{00}!
    N_{01}!
    N_{10}!
    N_{11}!
  }
  {
    (N_{0*}+1)!
    (N_{1*}+1)!
    (N_0/N)^{N_0}
    (N_1/N)^{N_1}
  }.
\end{equation}
The idea behind the lower benchmark is that,
for any power probability measure $Q^N$
($Q$ being a probability measure on $\{0,1\}$),
it is an e-variable w.r.\ to $Q^N$,
i.e., satisfies $\int\LB\d Q^N\le1$.
To ensure this,
\eqref{eq:LB} is defined as the ratio of the alternative probability measure
to the maximum likelihood under the IID model.

However, the IID model is not our null hypothesis,
and our null hypothesis of exchangeability is slightly more challenging.
Replacing in \eqref{eq:LB} the maximum likelihood over the IID model
by the maximum likelihood over the exchangeability model,
we obtain the \emph{exchangeability lower benchmark}
\begin{equation}\label{eq:ELB}
  \ELB
  :=
  \frac12
  \binom{N}{N_1}
  \frac{
    N_{00}!
    N_{01}!
    N_{10}!
    N_{11}!
  }
  {
    (N_{0*}+1)!
    (N_{1*}+1)!
  }.
\end{equation}
For the e-power of the exchangeability lower benchmark
we have the formula~\eqref{eq:super-expression}
with the second term $H(t_*Q)$ omitted.

To compute efficiently the likelihood ratio of the alternative to null probability measures,
we will use the following facts \cite[Lemmas 8.5 and 8.6]{Vovk/etal:2005book-local},
which are versions of standard results in graph theory
(the BEST theorem and the Matrix-Tree theorem).
We will use the terminology of \cite[Section 8.6]{Vovk/etal:2005book-local}
(such as ``Markov graph'')
and consider an arbitrary finite observation space $\mathbf{Z}$
(instead of $\{0,1\}$, as in the rest of this paper).

\begin{lemma}\label{lem:BEST}
  In any Markov graph $\sigma$
  with the set of vertices $V$
  the number of Eulerian paths from the source to the sink equals
  \begin{equation}\label{eq:BEST}
    T(\sigma)
    \frac
    {
      \outv(\sink)
      \prod_{v\in V}(\outv(v)-1)!
    }
    {
      \prod_{u,v\in V}N_{u,v}!
    },
  \end{equation}
  where $T(\sigma)$ is the number of spanning out-trees
  in the underlying digraph rooted at the source
  and $N_{u,v}$ is the number of darts leading from $u$ to $v$.
\end{lemma}

\begin{proof}
  According to Theorem~VI.28 in \cite{Tutte:1984}
  (and using the terminology of \cite[Chap.~VI]{Tutte:1984}),
  the number of Eulerian tours in the underlying digraph is
  \begin{equation*}
    T(\sigma)
    \prod_{v\in V}(\outv(v)-1)!.
  \end{equation*}
  If $\source=\sink$, the number of Eulerian paths
  is obtained by multiplying by $\outv(\source)$.
  Finally, we identify all darts from $u$ to $v$ for all pairs of vertices $(u,v)$
  by dividing by $N_{u,v}!$;
  the resulting expression agrees with \eqref{eq:BEST}.

  Now suppose $\source\ne\sink$.
  Create a new digraph by adding another dart leading from the source to the sink.
  The number of Eulerian paths from the source to the sink in the old digraph
  will be equal to the number of Eulerian tours in the new graph, i.e.,
  \begin{equation*}
    T(\sigma)
    \outv(\sink)
    \prod_{v\in V}(\outv(v)-1)!,
  \end{equation*}
  where $\outv$ refers to the old digraph.
  It remains to identify all darts from $u$ to $v$ for all pairs of vertices $(u,v)$
  in the old digraph;
  the resulting expression again agrees with \eqref{eq:BEST}.

  Alternatively, we can combine the two cases
  by always adding another dart leading from the source to the sink.
\end{proof}

\begin{lemma}\label{lem:matrix-tree}
  To find the number $T(\sigma)$ of spanning out-trees rooted at the source
  in the underlying digraph of a Markov graph $\sigma$ with vertices $z_1,\dots,z_n$
  ($z_1$ being the source),
  \begin{itemize}
  \item
    create the $n\times n$ matrix with the elements $a_{i,j}=-N_{z_i,z_j}$;
  \item
    change the diagonal elements so that each column sums to 0;
  \item
    compute the co-factor of $a_{1,1}$.
  \end{itemize}
\end{lemma}

\begin{proof}
  This lemma can be derived from Theorem~VI.28 in \cite{Tutte:1984}.
  In that theorem we can compute the co-factor of any diagonal element $a_{i,i}$,
  but it is about Eulerian digraphs.
  We can make the underlying digraph of our Markov graph Eulerian
  by connecting the sink to the source.
  This operation does not affect the number of out-trees rooted at the source
  and does not change the co-factor of $a_{1,1}$.
\end{proof}

Let us specialize Lemmas~\ref{lem:BEST} and~\ref{lem:matrix-tree}
to the binary case $\mathbf{Z}:=\{0,1\}$.

\begin{corollary}
  Let $\sigma$ be a Markov graph with vertices in $\{0,1\}$
  with $F$ as its source.
  The number of Eulerian paths from the source to the sink equals
  \begin{equation}\label{eq:N}
    N(\sigma)
    :=
    \begin{cases}
      N_{F,1-F}
      \frac{(N_0-1)!(N_1-1)!}{N_{00}!N_{01}!N_{10}!N_{11}!}
      & \text{if $N_0\wedge N_1>0$}\\
      1 & \text{otherwise},
    \end{cases}
  \end{equation}
  where $N_i:=\inv(i)\vee\outv(i)$
  and $N_{i,j}$ (with the comma omitted)
  is the number of darts leading from $i$ to $j$.
\end{corollary}

\begin{proof}
  The number of spanning out-trees rooted at the source
  in the underlying digraph is
  \[
    T(\sigma)
    =
    N_{F,1-F};
  \]
  this follows from Lemma~\ref{lem:matrix-tree}
  and is obvious anyway.
  It remains to plug this in into Lemma~\ref{lem:BEST}:
  assuming $N_0\wedge N_1>0$,
  if the source $F$ and sink $L$ coincide, $F=L$,
  we obtain
  \[
    N_{F,1-F}
    \frac{(N_F-1)(N_F-2)!(N_{1-F}-1)!}{N_{00}!N_{01}!N_{10}!N_{11}!},
  \]
  and if $F\ne L$, we obtain
  \[
    N_{F,1-F}
    \frac{(N_L-1)(N_F-1)!(N_L-2)!}{N_{00}!N_{01}!N_{10}!N_{11}!};
  \]
  both expression agree with \eqref{eq:N}.
  The case $N_0\wedge N_1=0$ is obvious.
\end{proof}

Combining \eqref{eq:alternative} and \eqref{eq:N},
we obtain the total alternative weight of
\begin{equation}\label{eq:W}
  W(\sigma)
  :=
  \begin{cases}
    \frac12
    N_{F,1-F}
    \frac{(N_0-1)!(N_1-1)!}{(N_{0*}+1)!(N_{1*}+1)!}
    & \text{if $N_{1-F}>0$}\\
    \frac{1}{2N} & \text{otherwise}
  \end{cases}
\end{equation}
for all data sequences of a given Markov type $\sigma$.

Under the null hypothesis the probability of a data sequence
of exchangeability type $(N_0,N_1)$ is
\[
  1 / \binom{N}{N_1},
\]
and so the likelihood ratio (the alternative over the null of exchangeability) is
\begin{equation}\label{eq:LR}
  \frac12
  \frac{
    N_{00}!
    N_{01}!
    N_{10}!
    N_{11}!
    \binom{N}{N_1}
  }
  {
    (N_{0*+1})!
    (N_{1*+2})!
    \sum_{\sigma}W(\sigma)
  }
  =
  \frac{
    N_{00}!
    N_{01}!
    N_{10}!
    N_{11}!
    \binom{N}{N_1}
  }
  {
    (N_{0*+1})!
    (N_{1*+2})!
    \sum_{\sigma}
    n_{f,1-f}
    \frac{(N_0-1)!(N_1-1)!}{(n_{0*}+1)!(n_{1*}+1)!}
  }
\end{equation}
(see \eqref{eq:alternative} and \eqref{eq:W}),
where the $\sigma$ in $\sum_{\sigma}$ ranges over the Markov types
$(f,n_{00},\dots,n_{11},l)$
compatible with the exchangeability type $(N_0,N_1)$.
The equality in \eqref{eq:LR} holds when $N_0\wedge N_1=0$;
in the case $N_0\wedge N_1=0$ the likelihood value is 1
(and we will treat this case separately in Algorithm~\ref{alg:UMM}).
We will refer to \eqref{eq:LR}
(interpreted as 1 when $N_0\wedge N_1=0$)
as the \emph{uniformly mixed Markov} (\emph{UMM}) e-variable;
this is our main object of interest in this paper.

It remains to explain how to compute the sum $\sum_{\sigma}$ in \eqref{eq:LR}.
For the $\sigma=(f,n_{00},\dots,n_{11},l)$ with $f=l=0$
(which is only possible when $N_0\ge2$),
each addend in the sum is
\[
  n_{f,1-f}
  \frac{(N_0-1)!(N_1-1)!}{(n_{0*}+1)!(n_{1*}+1)!}
  =
  n_{01}
  \frac{(N_0-1)!(N_1-1)!}{N_0!(N_1+1)!}
  =
  \frac{n_{01}}{N_0 N_1 (N_1+1)}.
\]
A specific Markov type $(f,n_{00},\dots,n_{11},l)$ is determined
(once we know that $f=l=0$) by $n_{01}$,
and its other components can be found from the equalities
\[
  \left\{
    \begin{aligned}
      n_{01} &= n_{10}\\
      N_0 &= n_{00}+n_{01}+1\\
      N_1 &= n_{01}+n_{11}.
    \end{aligned}
  \right.
\]
The valid values for $n_{01}$ are between $1$ and $(N_0-1)\wedge N_1$,
and so the part of the sum $\sum_{\sigma}$ corresponding to such $\sigma$ is
\begin{equation}\label{eq:FL=00}
  \sum_{n_{01}=1}^{(N_0-1)\wedge N_1}
  \frac{n_{01}}{N_0 N_1 (N_1+1)}
  =
  \frac
  {
    ((N_0-1)\wedge N_1)
    ((N_0-1)\wedge N_1 + 1)
  }
  {2 N_0 N_1 (N_1+1)}.
\end{equation}
This component should only be used when $N_0\ge2$;
otherwise, it is $0$.

For the $\sigma$ with $f=0$ and $l=1$,
the part of the sum $\sum_{\sigma}$ corresponding to such $\sigma$ is
\begin{equation}\label{eq:FL=01}
  \sum_{n_{01}=1}^{N_0\wedge N_1}
  \frac{n_{01}}{N_0 (N_0+1) N_1}
  =
  \frac
  {
    (N_0\wedge N_1)
    (N_0\wedge N_1 + 1)
  }
  {2 N_0 (N_0+1) N_1}.
\end{equation}
For the $\sigma$ with $f=1$ and $l=0$,
the part of the sum $\sum_{\sigma}$ corresponding to such $\sigma$ is
\begin{equation}\label{eq:FL=10}
  \sum_{n_{10}=1}^{N_0\wedge N_1}
  \frac{n_{10}}{N_0 N_1 (N_1+1)}
  =
  \frac
  {
    (N_0\wedge N_1)
    (N_0\wedge N_1 + 1)
  }
  {2 N_0 N_1 (N_1+1)}.
\end{equation}
Finally,
for the $\sigma$ with $f=l=1$,
the part of the sum $\sum_{\sigma}$ corresponding to such $\sigma$ is
\begin{equation}\label{eq:FL=11}
  \sum_{n_{10}=1}^{N_0\wedge(N_1-1)}
  \frac{n_{10}}{N_0 N_1 (N_1+1)}
  =
  \frac
  {
    (N_0\wedge(N_1-1))
    (N_0\wedge(N_1-1) + 1)
  }
  {2 N_0 (N_0+1) N_1}.
\end{equation}
This component is used only when $N_1\ge2$;
otherwise, we set it to $0$.

\begin{algorithm}[bt]
  \caption{Computing the UMM exchangeability e-variable}
  \label{alg:UMM}
  \begin{algorithmic}[1]
    \Require $(z_1,\dots,z_N)\in\{0,1\}^N$.
    \Ensure the value of the UMM e-variable $\UMM(z_1,\dots,z_N)$.
    \State Set $N_0$ and $N_1$ to the numbers of $0$s and $1$s in $(z_1,\dots,z_N)$,
        respectively.
    \If{$N_0\wedge N_1=0$:}
      \textbf{return} $1$
    \EndIf
    \For{$i,j\in\{0,1\}$:}
      \State Set $N_{i,j}$ to the number of substrings $(i,j)$
          in $(z_1,\dots,z_N)$.
    \EndFor
    \State $\ELB := \frac{N_{00}!N_{01}!N_{10}!N_{11}!\binom{N}{N_1}}
        {(N_{0*+1})!(N_{1*+2})!}$.\label{ln:BB-initial}
    \State $\Sum := 0$.\label{ln:Sum-initial}
    \If{$N_0\ge2$:}
      $\Sum \peq \frac{((N_0-1)\wedge N_1)((N_0-1)\wedge N_1 + 1)}
          {2 N_0 N_1 (N_1+1)}$.\label{ln:FL=00}
    \EndIf
    \State $\Sum \peq \frac{(N_0\wedge N_1)(N_0\wedge N_1 + 1)}
      {2 N_0 (N_0+1) N_1}$.\label{ln:FL=01}
    \State $\Sum \peq \frac{(N_0\wedge N_1)(N_0\wedge N_1 + 1)}
      {2 N_0 N_1 (N_1+1)}$.\label{ln:FL=10}
    \If{$N_1\ge2$:}
      $\Sum \peq \frac{(N_0\wedge(N_1-1))(N_0\wedge(N_1-1) + 1)}
        {2 N_0 (N_0+1) N_1}$.\label{ln:FL=11}
    \EndIf
    \State \textbf{return} $\ELB/\Sum$.
  \end{algorithmic}
\end{algorithm}

The overall algorithm is presented as Algorithm~\ref{alg:UMM}.
The value of the uniformly mixed Markov e-variable $\UMM$
is computed according to \eqref{eq:LR},
and the value $\ELB$ of the exchangeability lower benchmark in line~\ref{ln:BB-initial}
is just \eqref{eq:LR} with the sum over the Markov types $\sigma$ omitted.
In line~\ref{ln:Sum-initial} we initialize the sum over the Markov types $\sigma$,
and in lines \ref{ln:FL=00}, \ref{ln:FL=01}, \ref{ln:FL=10}, and \ref{ln:FL=11}
we compute it according to the right-hand sides of
\eqref{eq:FL=00}, \eqref{eq:FL=01}, \eqref{eq:FL=10}, and \eqref{eq:FL=11},
respectively.
The symbol $\peq$ is used in the Python sense:
$A \peq B$ is equivalent to $A:=A+B$.
The output is returned by the \textbf{return} command,
and the algorithm stops as soon as the first such command is issued.

The computational complexity of Algorithm~\ref{alg:UMM} is $O(N)$ time,
which is clearly optimal.

\section{Maximum e-power of the UMM alternative}
\label{sec:e-power}

In this section we will compute the asymptotic efficiency of the UMM e-variable
under the UMM alternative.
In the next section we will see the weakness of the notion of efficiency:
it has a long-run frequency interpretation,
but the logarithm of the UMM e-variable can be extremely volatile
(and so its mathematical expectation can be very different
from what we actually expect to observe).

\begin{proposition}\label{prop:ep}
  Under the UMM alternative $Q$,
  the asymptotic e-power of the UMM e-variable $\UMM$ (for horizon $N$) satisfies
  \[
    \lim_{N\to\infty}
    \mep(Q)/N
    =
    \lim_{N\to\infty}
    \ep_Q(\UMM)/N
    =
    \frac83 \ln 2
    +
    \frac23 \ln^2 2
    -
    \frac{7}{36} \pi^2
    -
    \frac16
    \approx
    0.083.
  \]
  The same expression gives the asymptotic e-power
  of the exchangeability lower benchmark
  (and of the lower benchmark).
\end{proposition}

\begin{proof}
  Let us compute separately the three components in \eqref{eq:super-expression},
  starting from the last one.

  When estimating $-H(Q)$,
  we need to estimate the frequencies
  $N_{00}$, $N_{01}$, $N_{10}$, $N_{11}$
  for a Markov chain with transition probabilities $\pi_{i,j}$.
  To this end, we define a new Markov chain
  whose states are the pairs $z_i z_{i+1}$, $i=1,\dots,N-1$,
  of adjacent states of the old chain with the matrix of transition probabilities
  \[
    P
    :=
    \begin{pmatrix}
      \pi_{00} & \pi_{01} & 0 & 0\\
      0 & 0 & \pi_{10} & \pi_{11}\\
      \pi_{00} & \pi_{01} & 0 & 0\\
      0 & 0 & \pi_{10} & \pi_{11}
    \end{pmatrix};
  \]
  the rows and columns of this matrix are labelled
  by the states $00$, $01$, $10$, and $11$ of the new Markov chain,
  in this order.
  The stationary probabilities for this $4\times4$ matrix are
  \[
    \left(
      \frac{\pi_{00}\pi_{10}}{\pi_{01}+\pi_{10}},
      \frac{\pi_{01}\pi_{10}}{\pi_{01}+\pi_{10}},
      \frac{\pi_{01}\pi_{10}}{\pi_{01}+\pi_{10}},
      \frac{\pi_{01}\pi_{11}}{\pi_{01}+\pi_{10}}
    \right).
  \]
  Now, assuming that the observations are generated
  from a Markov chain with transition probabilities $\pi_{i,j}$,
  we obtain (cf.~\eqref{eq:alternative})
  \begin{align*}
    \E&\ln
    \left(
      \frac12
      \frac{
        N_{00}!
        N_{01}!
        N_{10}!
        N_{11}!
      }
      {
        (N_{0*}+1)!
        (N_{1*}+1)!
      }
    \right)\\
    &=
    \E
    \bigl(
      N_{00} \ln N_{00} - N_{00}
      +
      N_{01} \ln N_{01} - N_{01}\\
      &\quad+
      N_{10} \ln N_{10} - N_{10}
      +
      N_{11} \ln N_{11} - N_{11}\\
      &\quad-
      (N_{00}+N_{01}+1) \ln(N_{00}+N_{01}+1) - (N_{00}+N_{01}+1)\\
      &\quad-
      (N_{10}+N_{11}+1) \ln(N_{10}+N_{11}+1) - (N_{10}+N_{11}+1)
    \bigr)
    +
    O(N^{1/2})\\
    &=
    \E
    \biggl(
      N_{00} \ln\frac{N_{00}}{N_{00}+N_{01}}
      +
      N_{01} \ln\frac{N_{01}}{N_{00}+N_{01}}\\
      &\quad+
      N_{10} \ln\frac{N_{10}}{N_{10}+N_{11}}
      +
      N_{11} \ln\frac{N_{11}}{N_{10}+N_{11}}
    \biggr)
    +
    O(N^{1/2})\\
    &=
    N \frac{\pi_{00}\pi_{10}}{\pi_{01}+\pi_{10}}
    \ln\pi_{00}
    +
    N \frac{\pi_{01}\pi_{10}}{\pi_{01}+\pi_{10}}
    \ln\pi_{01}\\
    &\quad+
    N \frac{\pi_{01}\pi_{10}}{\pi_{01}+\pi_{10}}
    \ln\pi_{10}
    +
    N \frac{\pi_{01}\pi_{11}}{\pi_{01}+\pi_{10}}
    \ln\pi_{11}
    +
    O(N^{1/2})
  \end{align*}
  (we are ignoring the special cases such as $N_{00}=0$,
  which should be considered separately).
  To find the expectation under the Bayes mixture of the Markov model
  with the uniform prior on $(\pi_{01},\pi_{10})$,
  we integrate
  \begin{align}
    \int_0^1
    \int_0^1&
    \biggl(
      \frac{\pi_{00}\pi_{10}}{\pi_{01}+\pi_{10}}
      \ln\pi_{00}
      +
      \frac{\pi_{01}\pi_{10}}{\pi_{01}+\pi_{10}}
      \ln\pi_{01}\notag\\
      &\quad+
      \frac{\pi_{01}\pi_{10}}{\pi_{01}+\pi_{10}}
      \ln\pi_{10}
      +
      \frac{\pi_{01}\pi_{11}}{\pi_{01}+\pi_{10}}
      \ln\pi_{11}
    \biggr)
    \d\pi_{01}\d\pi_{10}\notag\\
    &=
    \frac23 \ln 2
    +
    \frac23 \ln^2 2
    -
    \frac19 \pi^2
    -
    \frac16
    \approx
    -0.481.\label{eq:Part-A}
  \end{align}

  Now let us estimate the first term
  \[
    \int
    \ln
    \left|
      t^{-1}(\sigma)
    \right|
    (t_*Q)(\dd\sigma)
  \]
  in \eqref{eq:super-expression}.
  Set $K:=\sigma$ (this is the number of 1s),
  and suppose the observations are generated from a Markov chain
  with given transition probabilities $\pi_{01}$ and $\pi_{10}$.
  We then have
  \begin{align*}
    \E
    \left(
      \ln\binom{N}{K}
    \right)
    &=
    \E
    \left(
      \ln\frac{N!}{K!(N-K)!}
    \right)
    =
    \E
    \left(
      \ln\frac{(N/e)^N}{\left(\frac{K}{e}\right)^K\left(\frac{N-K}{e}\right)^{N-K}}
    \right)
    +
    O(N^{1/2})\\
    &=
    \E
    \left(
      -K\ln\frac{K}{N} - (N-K)\ln\left(1-\frac{K}{N}\right)
    \right)
    +
    O(N^{1/2})\\
    &=
    -N\pi_1\ln\pi_1 - N\pi_0\ln\pi_0
    +
    O(N^{1/2}),
  \end{align*}
  where $\pi_0$ and $\pi_1$ are the stationary probabilities
  \[
    \pi_0
    :=
    \frac{\pi_{10}}{\pi_{01}+\pi_{10}}
    \text{ and }
    \pi_1
    :=
    \frac{\pi_{01}}{\pi_{01}+\pi_{10}}
  \]
  of the Markov chain.
  It remain to take the integral
  \begin{align}
    -\int_0^1
    \int_0^1&
    \left(
      \pi_0\ln\pi_0
      +
      \pi_1\ln\pi_1
    \right)
    \d\pi_{01}\d\pi_{10}
    =
    -2
    \int_0^1
    \int_0^1
    \left(
      \pi_0\ln\pi_0
    \right)
    \d\pi_{01}\d\pi_{10}\notag\\
    &=
    -2
    \int_0^1
    \int_0^1
    \left(
      \frac{\pi_{10}}{\pi_{01}+\pi_{10}}
      \ln\frac{\pi_{10}}{\pi_{01}+\pi_{10}}
    \right)
    \d\pi_{01}\d\pi_{10}\notag\\
    &=
    2\ln2
    -
    \frac{1}{12} \pi^2
    \approx
    0.564.\label{eq:Part-B}
  \end{align}

  The final term $H(t_*Q)$ in \eqref{eq:super-expression} can be ignored.
  Indeed, using the last expression in \eqref{eq:alternative},
  we can bound the probability $(t_*Q)(\{K\})$, for any $K\in\{1,\dots,N-1\}$,
  by 1 from above and by $1/(2N)$ from below:
  \begin{equation}\label{eq:Part-C}
    (t_*Q)(\{K\})
    \ge
    \frac12
    \frac{(N-K-1)!0!1!(K-1)!}{(N-K-1)!K!}
    =
    \frac{1}{2K}
    \ge
    \frac{1}{2N}
  \end{equation}
  (the expression after the first ``$\ge$'' being the probability of the sequence
  consisting of $K$ 1s followed by $N-K$ 0s).
  Therefore, $H(t_*Q)=O(\ln N)$.
  (As always, the extreme cases $K\in\{0,N\}$ should be considered separately.)

  Combining \eqref{eq:Part-A} and \eqref{eq:Part-B},
  we obtain the coefficient
  \begin{equation}\label{eq:combination}
    \frac83 \ln 2
    +
    \frac23 \ln^2 2
    -
    \frac{7}{36} \pi^2
    -
    \frac16
    \approx
    0.083
  \end{equation}
  in front of $N$ in the asymptotic expression for $\ep_Q(\UMM)$.

  The proof shows that the asymptotic e-power is the same
  for the exchangeability lower benchmark,
  and a simple calculation using Stirling's formula
  (see, e.g., \cite[Proposition 9.2]{Vovk/etal:2022book})
  shows that we also have the same asymptotic e-power
  for the lower benchmark.
\end{proof}

The proof of Proposition~\ref{prop:ep} contains
the following relation between the UMM e-variable and the exchangeability lower benchmark;
in particular, it shows once again that the exchangeability lower benchmark is also an e-variable.

\begin{proposition}\label{prop:ELB}
  It is always true that
  \begin{equation}\label{eq:ELB-bound-loose}
    1
    \le
    \frac{\UMM}{\ELB}
    \le
    2N.
  \end{equation}
  Moreover, it is true that
  \begin{equation}\label{eq:ELB-bound}
    1
    \le
    \frac{\UMM}{\ELB}
    \le
    N
  \end{equation}
  unless $N_1\in\{0,N\}$.
\end{proposition}

\begin{proof}
  Unless $N_1=0$,
  we can improve \eqref{eq:Part-C} to
  \begin{equation*}
    (t_*Q)(\{K\})
    \ge
    \frac{1}{N}
  \end{equation*}
  by considering,
  alongside the sequence consisting of $K$ 1s followed by $N-K$ 0s,
  the sequence consisting of $K$ 0s followed by $N-K$ 1s.
  It remains to notice that
  \[
    \UMM
    =
    \frac{\ELB}{(t_*Q)(\{N_1\})}.
    \qedhere
  \]
\end{proof}

\section{Computational experiments}
\label{sec:experiments}

We will conduct two groups of experiments
for the two lower benchmarks and the UMM exchangeability e-variable.
In the first group, the true data distribution
will be a specific Markov probability measure
with initial probability of 1 equal to 1/2.
In this case, we define the \emph{upper benchmark} as
\begin{equation}\label{eq:UB}
  \UB
  :=
  \frac12
  \frac{
    N_{00}!
    N_{01}!
    N_{10}!
    N_{11}!
  }
  {
    (N_{0*}+1)!
    (N_{1*}+1)!
    \pi_0^{N_0}
    \pi_1^{N_1}
  },
\end{equation}
where $\pi_0$ and $\pi_1$ are the stationary probabilities
under the true data-generating distribution.
Therefore, the upper benchmark is an e-variable
w.r.\ to a specific IID probability measure,
and so it is not even an IID e-variable.
Therefore, we should not be surprised if the upper benchmark
exceeds a bona fide exchangeability e-variable;
there are two elements of cheating in interpreting the upper benchmark
as measure of evidence against the null hypothesis of exchangeability:
first, it tests IID rather than exchangeability,
and second, it tests only one individual IID measure.

\begin{figure}[b]
  \begin{center}
    \includegraphics[width=0.48\textwidth]{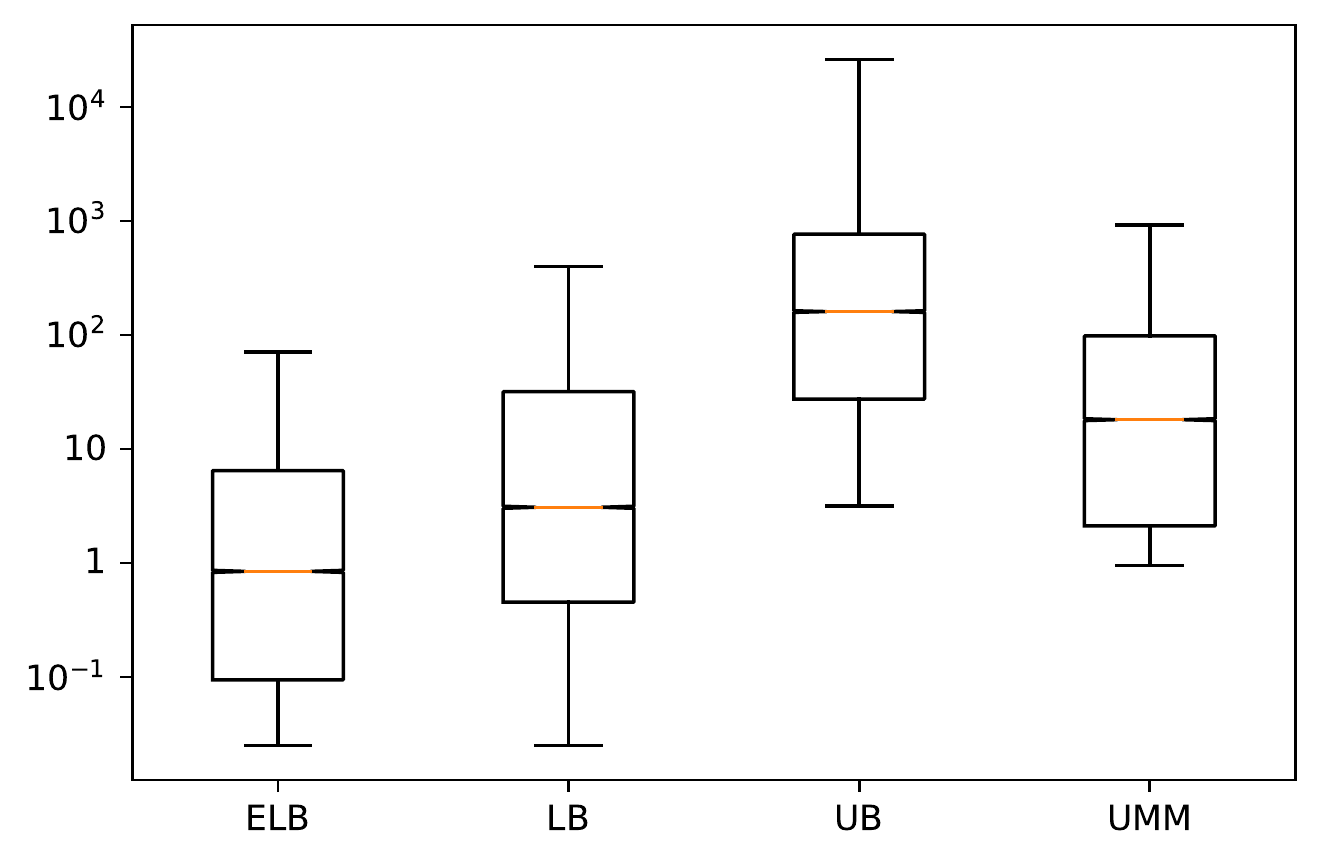}
    \includegraphics[width=0.48\textwidth]{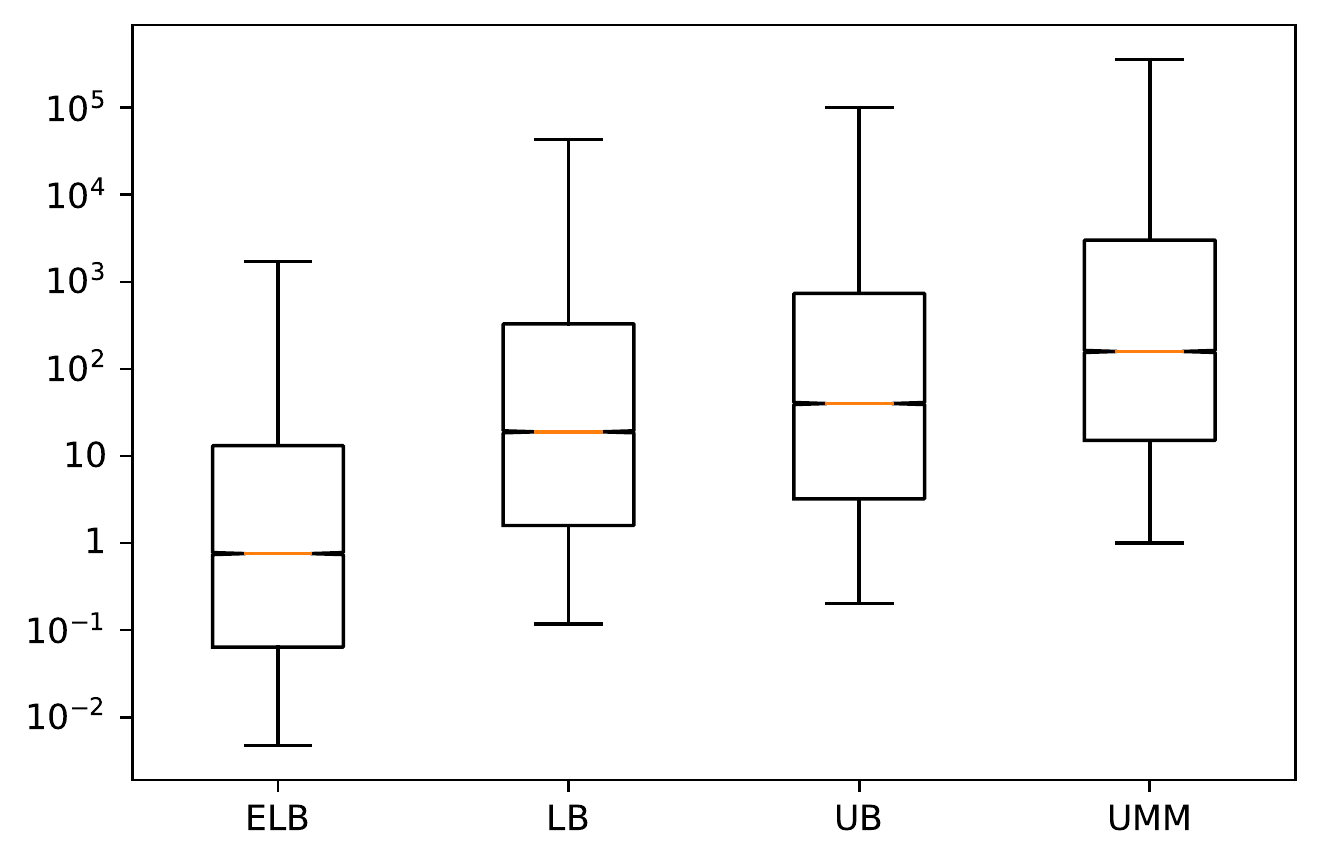}
  \end{center}
  \caption{The four e-values and related quantities, as described in text.
    Left panel: $N=20$ and $\pi_{01}=\pi_{10}=0.1$.
    Right panel: $N=400$ and $\pi_{01}=\pi_{10}=0.4$.}
  \label{fig:specific}
\end{figure}

Our results for specific Markov alternatives are given in Fig.~\ref{fig:specific}.
This figure contains boxplots for $10^5$ simulations of four values:
the exchangeability lower benchmark $\ELB$ (given by \eqref{eq:ELB}),
the lower benchmark $\LB$ (given by \eqref{eq:LB}),
the upper benchmark $\UB$ (given by \eqref{eq:UB}),
and the UMM exchangeability e-variable $\UMM$ (given by Algorithm~\ref{alg:UMM}).
Only two of these, $\ELB$ and $\UMM$, are bona fide exchangeability e-values.
It is interesting that $\UMM$ is often even higher than the upper benchmark,
as in the right panel of Fig.~\ref{fig:specific}.
The horizon $N$ and the transition probabilities
for the two panels are given in the caption.
In both cases, the alternative probability measure is Markov.

\begin{table}
  \begin{center}
    {\footnotesize
    \begin{tabular}{cccccc}
      $N$ & $\pi_{01}=\pi_{10}$ & $\overline{\text{ELB}}$ & $\overline{\text{UMM}}$ & $\overline{\text{UMM}}-\overline{\text{ELB}}$ & upper bound \\
      \hline\\[-2.8mm]
      20 & 0.1 & $-0.116$ & 1.226 & 1.342 & 1.301 \\
      400 & 0.4 & 0.084 & 2.427 & 2.343 & 2.602
    \end{tabular}}
  \end{center}
  \caption{Comparison between the decimal logarithms
    of the exchangeability lower benchmark and the UMM e-variable;
    the upper bound for the difference is $\log_{10}N$, as per \eqref{eq:ELB-bound}.}
  \label{tab:comparison-ergodic}
\end{table}

According to Proposition~\ref{prop:ELB},
the UMM e-value cannot differ from the exchangeability lower benchmark by much.
Table~\ref{tab:comparison-ergodic} gives the means of their decimal logarithms
(over the same $10^5$ simulations as in Fig.~\ref{fig:specific})
and the upper bound \eqref{eq:ELB-bound} for the difference between them.
The bars stand for the empirical averages (over all $10^5$ replications).
The upper bound \eqref{eq:ELB-bound} is violated in the first row
because $\pi_{01}=\pi_{10}$ and $N$ are so small,
which often leads to $N_1\in\{0,N\}$;
of course, the upper bound \eqref{eq:ELB-bound-loose}
(whose value is approximately 1.602 in this case) still holds.

\begin{figure}[b]
  \begin{center}
    \includegraphics[width=0.48\textwidth]{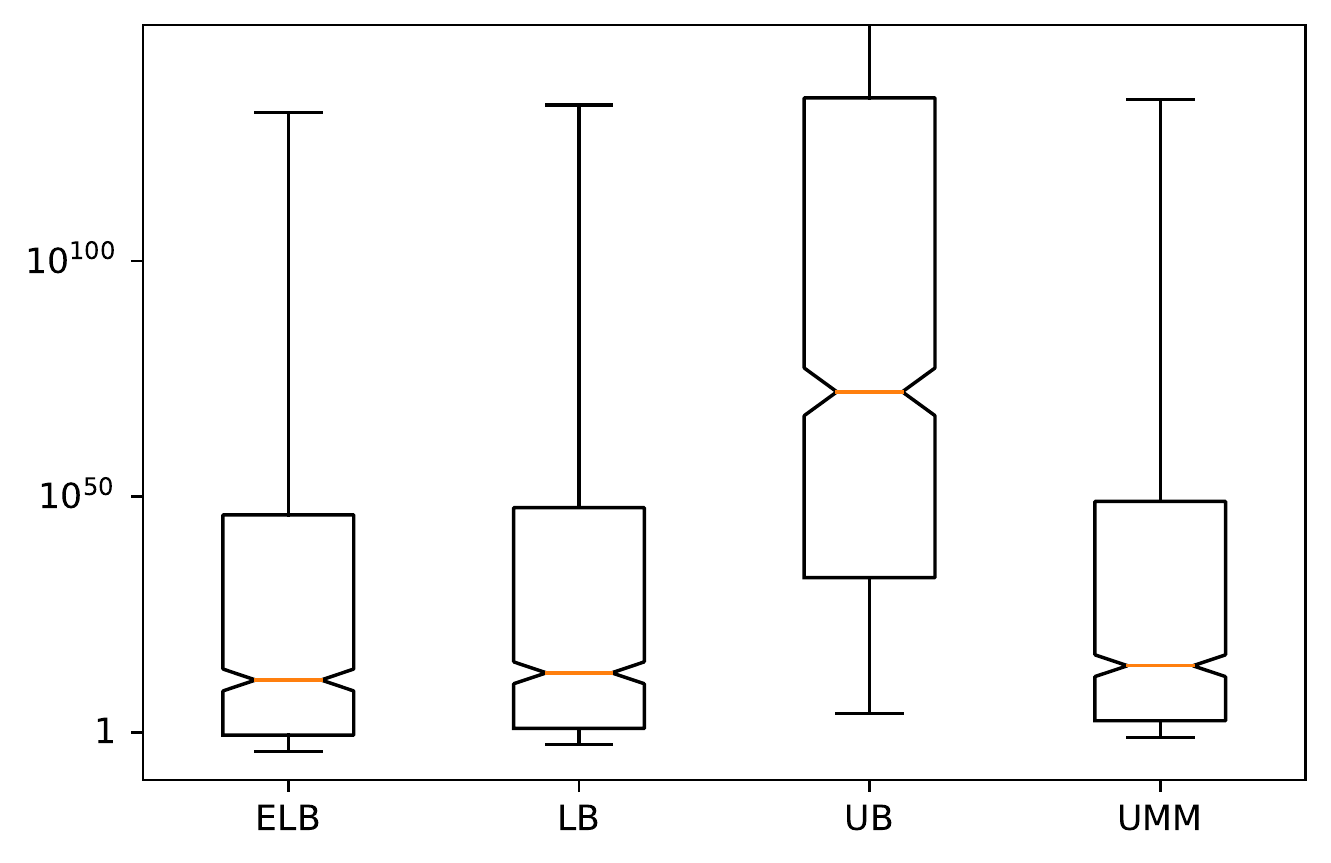}
    \includegraphics[width=0.48\textwidth]{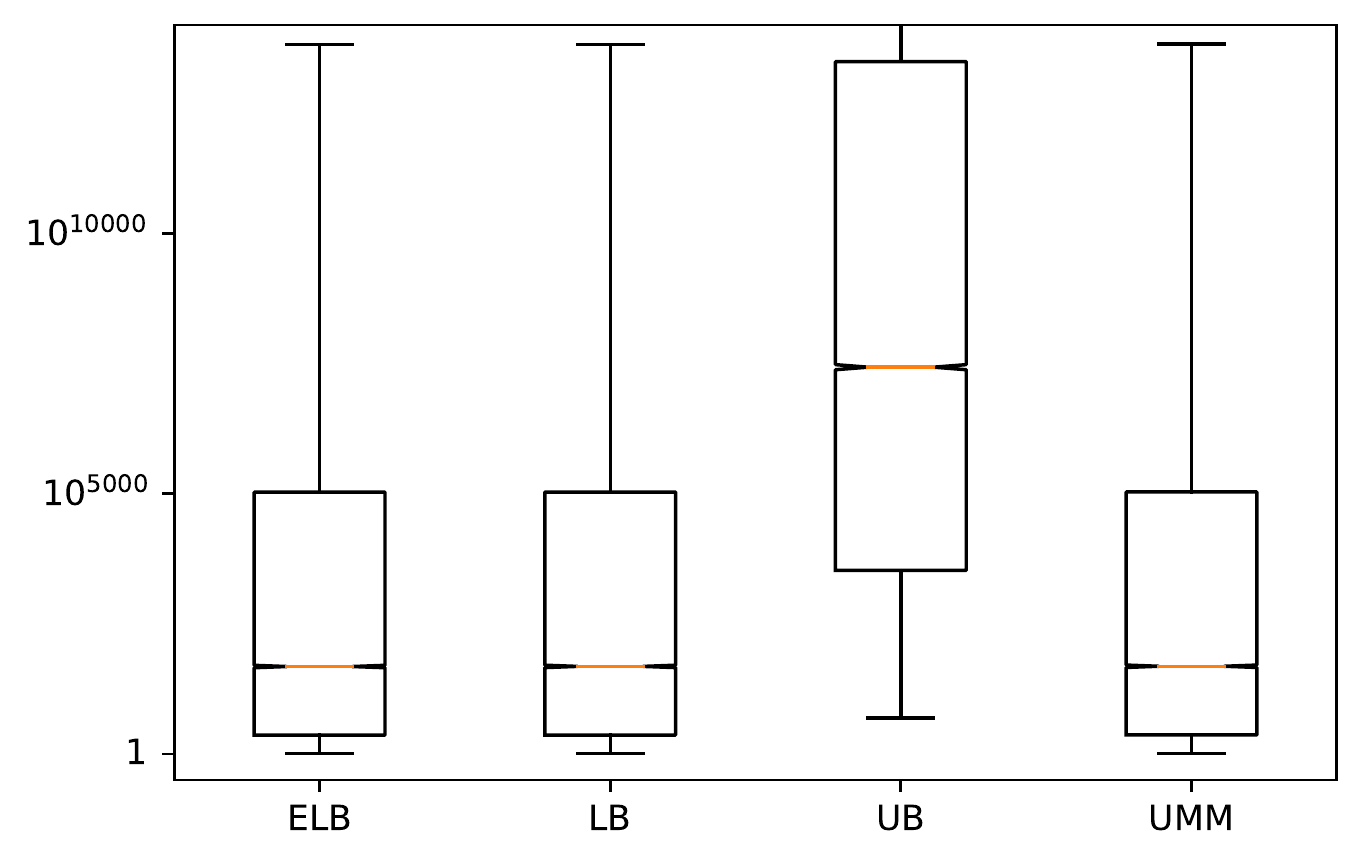}
  \end{center}
  \caption{Two exchangeability e-values and two approximations for the UMM alternative.
    Left panel: $N=K=10^3$.
    Right panel: $N=K=10^5$.}
  \label{fig:mixture}
\end{figure}

The second group of experiments involves generating the binary observations
from the UMM alternative (which is not Markov any more).
The explicit formula for this alternative is given in \eqref{eq:alternative},
but it is easier to generate $\pi_{01}$ and $\pi_{10}$ from the uniform distribution on $[0,1]^2$
and then generate the observations from the Markov chain with these parameters.
Figure~\ref{fig:mixture} shows results for this case,
with the same four values as in Fig.~\ref{fig:specific};
in the expression \eqref{eq:UB} for the upper benchmark,
we now set $\pi_{01}:=\pi_{10}:=1/2$.
It is striking how spread out the distributions for the three benchmarks
and the UMM e-variable are.
They are also skewed, with the mean very different from the median.
Now the lack of validity for the upper benchmark is very obvious:
it takes much larger values, and we will ignore it from now on.

\begin{table}
  \begin{center}
    {\footnotesize
    \begin{tabular}{ccccccc}
      $N$ & $K$ & $\overline{\text{ELB}}$ & $\overline{\text{LB}}$ & $\overline{\text{UMM}}$ & as.\ & UMM quantiles \\
      \hline\\[-2.8mm]
      $10^3$ & $10^3$ & 31.05 & 32.56 & 34.02 & 36.02 & [$-1.01$, 2.56, 14.21, 49.04, 134.22] \\
      $10^4$ & $10^4$ & 356.8 & 358.8 & 360.8 & 360.2 & [0.1, 35.7, 170.0, 509.8, 1378.3] \\
      $10^5$ & $10^5$ & 3570 & 3573 & 3575 & 3602 & [12, 366, 1684, 5033, 13632]
    \end{tabular}}
  \end{center}
  \caption{Some figures for the decimal logarithms of the two lower benchmarks and the UMM e-variable;
    ``as.''\ stands for the asymptotic expression 
    The UMM quantiles are for 5\%, 25\% (first quartile), 50\% (median), 75\% (third quartile), and 95\%.}
  \label{tab:mixture}
\end{table}

Table~\ref{tab:mixture} gives corresponding figures.
Now the bars stand for the empirical averages over $K$ replications
(for three different values of $K$),
$N$ is the horizon,
and ``as.''\ is the common theoretical asymptotic value
for the UMM e-variable and exchangeability lower benchmark obtained from \eqref{eq:combination}
by dividing by $\ln10$ (to convert natural logarithms to decimal ones)
and multiplying by the sample size $N$.

\begin{table}
  \begin{center}
    {\footnotesize
    \begin{tabular}{cccccc}
      $N$ & $K$ & $\overline{\text{ELB}}$ & $\overline{\text{UMM}}$ & $\overline{\text{UMM}}-\overline{\text{ELB}}$ & upper bound \\
      \hline\\[-2.8mm]
      $10^3$ & $10^3$ & 31.05 & 34.02 & 2.968 & 3 \\
      $10^4$ & $10^4$ & 356.8 & 360.8 & 3.964 & 4 \\
      $10^5$ & $10^5$ & 3570 & 3575 & 4.966 & 5
    \end{tabular}}
  \end{center}
  \caption{Analogue of Table~\ref{tab:comparison-ergodic}
    in the situation of Fig.~\ref{fig:mixture}.}
  \label{tab:comparison}
\end{table}

Table~\ref{tab:comparison}, analogous to Table~\ref{tab:comparison-ergodic},
gives the average differences
between the UMM e-variable and exchangeability lower benchmark on the $\log_{10}$ scale,
together with the upper bound given by~\eqref{eq:ELB-bound}.
The upper bound is surprisingly tight.

\section{Conclusion}
\label{sec:conclusion}

In this paper the algorithm for computing the UMM e-variable
was fully developed only in the binary case.
A natural next step would be to extend it to any finite observation space $\mathbf{Z}$.
(A big chunk of Sect.~\ref{sec:algorithm},
following \cite[Sect.~8.6]{Vovk/etal:2005book-local},
presented the combinatorics
for an arbitrary finite observation space $\mathbf{Z}$.)
It is interesting what the computational complexity
of such an extension of Algorithm~\ref{alg:UMM} will be
in general as function of $N$ and $\lvert\mathbf{Z}\rvert$.

The topic of this paper has been testing the exchangeability model in the batch mode
using Markov alternatives.
There are many other interesting null hypotheses among Kolmogorov compression models,
and there are many interesting alternatives.
For example, in \cite[Chap.~9]{Vovk/etal:2022book} we discussed,
alongside Markov alternatives,
detecting changepoints.
Our discussion was in the online mode,
but for changepoint detection the batch mode is not less important
\cite[Remark 8.19]{Vovk/etal:2022book};
e.g., its role has increased in bioinformatics
(including DNA analysis).
Using e-values in changepoint detection is particularly convenient
when multiple hypothesis testing is involved
(as it often is in batch changepoint detection).
Some extensions will be discussed in Appendixes~\ref{app:universal}--\ref{app:changepoint},
including changepoint detection in Appendix~\ref{app:changepoint}.

\appendix
\section{Algorithmic theory of randomness}
\label{app:ATR}

In the main part of the paper we avoided using computability theory
that plays such an important role in Kolmogorov's original approach,
which is the topic of this appendix.
Kolmogorov's complexity models were introduced, in the most complete form,
in what appears to be Kolmogorov's last talk,
given on 14 October 1982 at what later became known as the Kolmogorov seminar;
see \cite[note 12]{Semenov/etal:arXiv2303},
containing Shen's notes taken during the talk,
and \cite[Sect.~4]{Vovk:2001Denmark}.
The Kolmogorov seminar at Moscow State University
was opened by Kolmogorov on 28 October 1981,
and Kolmogorov gave two talks in it,
on 26 November 1981 and 14 October 1982
\cite[note 12]{Semenov/etal:arXiv2303};
the two talks were conflated in my paper \cite[Sect.~4]{Vovk:2001Denmark}.

All results listed in this appendix are either well known
or immediately follow from well-known results.

\subsection*{Mathematical results}

Let $X$ be an infinitely countable set
with a fixed bijection $f:X\to\{1,2,\dots\}$ between $X$ and the natural numbers.
When talking about the computability of functions involving elements $X$,
we mean the computability of those functions with $x\in X$ replaced
by their ``codes'' $f(x)$.
An example (of primary interest to us in this paper)
is the set of all finite binary sequences
with a computable bijection $f$.
Alternatively, $X$ can be any aggregate of constructive objects
in the sense of \cite[Sect.~1.0.6]{Uspensky/Semenov:1993};
in general, we will regard elements of $X$ as constructive objects.

Let us use the notation $C(x)$ for the Kolmogorov complexity of $x$,
$C(x\mid y)$ for the conditional Kolmogorov complexity of $x$ given $y$,
$K(x)$ for the prefix complexity of $x$,
and $K(x\mid y)$ for the conditional prefix complexity of $x$ given $y$.
Here $x$ and $y$ are constructive objects,
such as elements of $X$ or finite subsets of $X$.

Kolmogorov's definition of randomness deficiency
of an element $x$ of a finite set $A\subset X$ is
\begin{equation}\label{eq:Kolmogorov}
  d_A^C(x)
  :=
  \log\lvert A\rvert
  -
  C(x\mid A)
\end{equation}
\cite[Sect.~2.3]{Kolmogorov/Uspensky:1987-Latin},
where $\log$ stands for binary logarithm.
Informally, $x$ is random in $A$ if $d_A^C(x)$ is small.
(And Kolmogorov called $x$ $\Delta$-random in $A$
if $d_A^C(x)\le\Delta$.)

Martin-L\"of \cite{Martin-Lof:1966}
showed that Kolmogorov's definition \eqref{eq:Kolmogorov}
can be stated in terms of p-values.
Let $A$ be a finite subset of $X$
and $U_A$ be the uniform probability measure on $A$.
A function $f:A\to[0,1]$ is a \emph{p-variable} if
\[
  \forall\epsilon>0:
  U_A(f\le\epsilon)
  \le
  \epsilon.
\]
A family $P$ of functions $P_A:A\to[0,1]$,
$A$ ranging over the finite subsets of $X$,
is a \emph{p-test} if
\begin{itemize}
\item
  the function $(A,x)\mapsto P_A(x)$ is upper semicomputable,
  i.e., there is an algorithm that eventually stops
  on input $(A,x,\epsilon)$, where $\epsilon$ is a rational number,
  if and only if $P_A(x)<\epsilon$, and
\item
  for each finite $A\subset X$,
  $P_A$ is a p-variable.
\end{itemize}
The values taken by p-variables are \emph{p-values}.

\begin{lemma}\label{lem:p-universal}
  There exists a \emph{universal} p-test $\tilde P$,
  in the sense that for any p-test $P$ there exists a positive constant $c$
  such that $\tilde P \le c P$.
\end{lemma}

The proof of Lemma~\ref{lem:p-universal} is standard
(cf., e.g., \cite[Theorem 39]{Shen/etal:2017book}).
Fix a universal p-test $\tilde P$.
The universal p-test is unique to within a constant factor,
and it is customary in the algorithmic theory of randomness
to disregard such differences,
which we will also do in this appendix.

\begin{remark}
  The usual definitions in the algorithmic theory of randomness
  are given in terms of $-\log P$,
  but for simplicity let us discard the logarithm,
  following \cite{Vovk/Vyugin:1993}.
\end{remark}

Now we can state Martin-L\"of's result
expressing Kolmogorov's deficiency of randomness
via the universal p-test.
In this appendix $\log$ always stands for base 2 logarithm.

\begin{proposition}\label{prop:C}
  There exists a constant $c>0$ such that,
  for all $A$ and $x$,
  \begin{equation}\label{eq:C}
    \left|
      d^C_A(x)
      +
      \log\tilde P_A(x)
    \right|
    \le
    c.
  \end{equation}
\end{proposition}

\begin{proof}
  Martin-L\"of states and proves a slightly less general result
  in \cite[Sect.~II, Theorem on p.~607]{Martin-Lof:1966}
  (see also \cite[Sect.~V, Theorem on p.~616]{Martin-Lof:1966}),
  but his argument is general.
  Since, for each finite set $A\subset X$ and each $n\in\{0,1,\dots\}$,
  we have
  \[
    \left|\left\{
      x\in A
      \mid
      C(x\mid A) \le n
    \right\}\right|
    \le
    2^{n+1},
  \]
  we will also have
  \[
    U_A
    \left(\left\{
      x\in A
      \mid
      \log\lvert A\rvert - C(x\mid A) \ge n
    \right\}\right)
    \le
    2^{-n+2},
  \]
  which implies the part
  \[
    d^C_A(x)
    +
    \log\tilde P_A(x)
    \le
    c
  \]
  of \eqref{eq:C}.

  To prove the other part of \eqref{eq:C},
  i.e.,
  \[
    C(x\mid A)
    \le
    \log\lvert A\rvert
    +
    \log\tilde P_A(x)
    +
    c,
  \]
  it suffices to establish that,
  for some $c$ (perhaps a different one),
  \[
    \forall(A,x):
    \left|\left\{
      x\in A
      \mid
      \log\lvert A\rvert + \log\tilde P_A(x) \le n
    \right\}\right|
    \le
    2^{-n+c}.
  \]
  The last inequality follows immediately from the definition of a p-test
  (with $c:=0$).
\end{proof}

Prefix complexity $K$ has important technical advantages over $C$,
and so a natural modification of \eqref{eq:Kolmogorov} is
\begin{equation}\label{eq:prefix}
  d^K_A(x)
  :=
  \log\lvert A\rvert
  -
  K(x\mid A).
\end{equation}
Analogously to expressing \eqref{eq:Kolmogorov} in terms of p-values,
we can express \eqref{eq:prefix} in terms of e-values.

A function $f:A\to[0,\infty)$ on a finite set $A\subset X$
is an \emph{e-variable} if
\[
  \int f \d U_A
  \le
  1.
\]
A family $E$ of functions $E_A:A\to[0,1]$,
$A$ ranging over the finite subsets of $X$,
is an \emph{e-test} if
\begin{itemize}
\item
  the function $(A,x)\mapsto E_A(x)$ is \emph{lower semicomputable},
  i.e., there is an algorithm that eventually stops
  on input $(A,x,t)$, where $t$ is a rational number,
  if and only if $E_A(x)>t$, and
\item
  for each finite $A\subset X$,
  $E_A$ is an e-variable.
\end{itemize}

\begin{lemma}\label{lem:e-universal}
  There exists a \emph{universal} e-test $\tilde E$,
  in the sense that for any e-test $E$ there exists a positive constant $c$
  such that $\tilde E \ge E/c$.
\end{lemma}

The proof of Lemma~\ref{lem:e-universal} is again standard
(but \cite[Theorem 47]{Shen/etal:2017book} is now more relevant).
Fix a universal e-test $\tilde E$.
It is clear that the universal e-test is unique to within a constant factor.

Notice the difference between the universal tests
in Lemma~\ref{lem:p-universal} and Lemma~\ref{lem:e-universal}:
whereas in the former ``universal'' means ``smallest'' (to within a constant factor),
in the latter ``universal'' means ``largest''.
The following result expresses the prefix version \eqref{eq:prefix} of deficiency of randomness
via the universal e-test.

\begin{proposition}\label{prop:K}
  There exists a constant $c>0$ such that,
  for all $A$ and $x$,
  \begin{equation}\label{eq:K}
    \left|
      d^K_A(x)
      -
      \log\tilde E_A(x)
    \right|
    \le
    c.
  \end{equation}
\end{proposition}

Proposition~\ref{prop:K} will follow from two other propositions
(Propositions~\ref{prop:m-K} and~\ref{prop:m-E} below),
which, despite their simplicity (especially Proposition~\ref{prop:m-E}),
are of great independent interest.

A function $f:A\to[0,1]$ on a finite set $A\subset X$
is a \emph{subprobability measure}
(or \emph{semimeasure} \cite[Sect.~4.1]{Shen/etal:2017book})
if
\[
  \sum_{x\in A} f(x)
  \le
  1.
\]
A family $m$ of functions $m_A:A\to[0,1]$,
$A$ ranging over the finite subsets of $X$,
is a \emph{lower semicomputable subprobability measure} if
\begin{itemize}
\item
  the function $(A,x)\mapsto m_A(x)$ is lower semicomputable, and
\item
  for each finite $A\subset X$,
  $m_A$ is a subprobability measure.
\end{itemize}

\begin{lemma}\label{lem:m-universal}
  There exists a \emph{universal} lower semicomputable subprobability measure $\tilde m$,
  in the sense that for any lower semicomputable subprobability measure $m$
  there exists a positive constant $c$ such that $\tilde m \ge m/c$.
\end{lemma}

For a proof of, essentially, Lemma~\ref{lem:m-universal},
see the proof of \cite[Theorem 47]{Shen/etal:2017book}.
Let us abbreviate ``universal lower semicomputable subprobability measure''
to \emph{universal measure}.

\begin{proposition}\label{prop:m-K}
  There exists a constant $c>0$ such that,
  for all $A$ and $x$,
  \begin{equation*}
    \left|
      K(x\mid A)
      +
      \log\tilde m_A(x)
    \right|
    \le
    c.
  \end{equation*}
\end{proposition}

\begin{proof}
  Follow \cite[Sect.~4.5]{Shen/etal:2017book}.
\end{proof}

\begin{proposition}\label{prop:m-E}
  There exists a constant $c>0$ such that,
  for all $A$ and $x$,
  \begin{equation}\label{eq:m-E}
    \frac1c
    \le
    \frac{\tilde m_A(x)\lvert A\rvert}{\tilde E_A(x)}
    \le
    c.
  \end{equation}
\end{proposition}

\begin{proof}
  It suffices to notice that $\tilde m_A(x)\lvert A\rvert$
  is an e-test
  and that $\tilde E_A(x)/\lvert A\rvert$
  is a lower semicomputable subprobability measure.
\end{proof}

The interpretation of~\eqref{eq:m-E} is that the universal e-test $\tilde E$
is a likelihood ratio:
we divide the universal measure $\tilde m$ (``universal alternative hypothesis'')
by the null uniform probability measure $1/\lvert A\rvert$.

Now we can easily prove Proposition~\ref{prop:K}.

\begin{proof}[Proof of Proposition~\ref{prop:K}]
  Combining the previous propositions,
  we obtain
  \begin{equation}\label{eq:derivation}
  \begin{aligned}
    \left|
      d^K_A(x)
      -
      \log\tilde E_A(x)
    \right|
    &=
    \left|
      \log\lvert A\rvert - K(x\mid A)
      -
      \log\tilde E_A(x)
    \right|\\
    &\le
    \left|
      \log\lvert A\rvert + \log\tilde m_A(x)
      -
      \log(\tilde m_A(x)\lvert A\rvert)
    \right|
    +c
    =
    c,
  \end{aligned}
  \end{equation}
  i.e., \eqref{eq:K}.
  The first equality in \eqref{eq:derivation} just uses the definition of $d^K_A(x)$,
  and the inequality ``$\le$'' is obtained
  by applying Proposition~\ref{prop:m-K} to $K(x\mid A)$
  and applying Proposition~\ref{prop:m-E} to $\tilde E_A(x)$.
\end{proof}

Both complexities $C$ and $K$ and randomness deficiencies $d^C$ and $d^K$
are close to each other.

\begin{proposition}\label{prop:K-vs-C}
  There is a constant $c>0$ such that, for all finite $A\subset X$ and all $x\in A$,
  \begin{equation}\label{eq:K-vs-C}
    C(x\mid A) - c \le K(x\mid A) \le C(x\mid A) + 2 \log C(x\mid A) + c
  \end{equation}
  and
  \begin{equation}\label{eq:dK-vs-dC}
    d^K_A(x) - c \le d^C_A(x) \le d^K_A(x) + 2 \log d^K_A(x) + c.
  \end{equation}
\end{proposition}

\begin{proof}
  See \cite[Theorem 65]{Shen/etal:2017book} for inequalities stronger than \eqref{eq:K-vs-C}.
  For \eqref{eq:dK-vs-dC},
  follow the proof of \cite[Proposition 1]{Novikov:arXiv1608}.
\end{proof}

\subsection*{Discussion}

Kolmogorov's original definition of randomness deficiency
of an element of a finite set is \eqref{eq:Kolmogorov}.
It can be interpreted as the universal p-value
on the logarithmic scale (Proposition~\ref{prop:C}).
A natural modification of Kolmogorov's definition is \eqref{eq:prefix},
given in terms of prefix complexity,
and it can be interpreted as the universal e-value
on the logarithmic scale (Proposition~\ref{prop:K}).

The simplest context in which these definitions can be used
is that of \emph{complexity models},
in the terminology of \cite{Vovk:2001Denmark,Vovk/Shafer:2003}.
A complexity model is a computable partition of the sample space,
and the implicit statement about the observed data sequence $x$
is that it is random in the sense of \eqref{eq:Kolmogorov}
(or \eqref{eq:prefix}, which is close by Proposition~\ref{prop:K-vs-C})
in the block $A\ni x$ of the partition.
Let me give several examples of such models,
those that are most relevant in the context of this paper.
The sample space in all these examples will be $\{0,1\}^*$.
\begin{itemize}
\item
  The main complexity model of interest to Kolmogorov
  \cite{Kolmogorov:1968-Latin,Kolmogorov:1983-Latin}
  was that of \emph{exchangeability},
  where the binary sequences $\{0,1\}^*$
  are divided into the blocks of sequences of the same length
  and with the same number of 1s.
  Stripping this complexity model of the algorithmic theory of randomness,
  we obtain the exchangeability compression model
  introduced in the main part of the paper.
\item
  Another complexity model \cite{Kolmogorov:1983-Latin}
  is the Markov model,
  in which the blocks consist of the binary sequences
  with the identical first element
  and the same number of substrings $00$, $01$, $10$, and $11$.
  In the terminology of \cite[Sect.~11.3.4]{Vovk/etal:2022book},
  the exchangeability model is more specific than the Markov model.
\item
  A further generalization of the exchangeability complexity model
  is the second order Markov model
  (suggested in Kolmogorov's 1982 seminar talk \cite{Vovk:2001Denmark}),
  in which the blocks consist of the binary sequences
  with the identical first and second elements
  and the same number of substrings
  $000$, $001$, $010$, $011$, $100$, $101$, $110$, and $111$.
\item
  A model not considered by Kolmogorov is the \emph{changepoint model}
  (exchangeability with a changepoint),
  in which the blocks are indexed by $(N,\tau,K_0,K_1)$,
  where $N\in\{2,3,\dots\}$ (the horizon),
  $\tau\in\{1,\dots,N-1\}$ (the changepoint),
  $K_0\in\{0,\dots,\tau\}$, and $K_1\in\{0,\dots,N-\tau\}$, 
  and the block $(N,\tau,K_0,K_1)$ consists of all binary sequences of length $N$
  with $K_0$ 1s among their first $\tau$ elements and $K_1$ 1s among their last $N-\tau$ elements.
\end{itemize}

Other complexity models introduced by Kolmogorov were Gaussian and Poisson
(in his 1982 seminar talk \cite[note 12]{Semenov/etal:arXiv2303};
see also \cite{Asarin:1987,Asarin:1988short} and \cite[Sect.~4]{Vovk:2001Denmark}).
A complexity model formalizing the property of being IID rather than exchangeability
was introduced in work \cite{Vovk:1986-arXiv} done under Kolmogorov's supervision.

\subsection*{Stochastic sequences}

Kolmogorov's 1981 seminar talk was devoted
to what he called stochastic sequences,
which can be interpreted as an overarching structure over complexity models.
Let us say that a binary data sequence $x\in X$ is \emph{$(\alpha,\beta)$-stochastic}
if there is a finite set $A\subset X$
such that $C(A)\le\alpha$ and $d_A^C(x)\le\beta$.
And let us say that $x\in X$ is \emph{$\Delta$-random}
w.r.\ to a Kolmogorov complexity model if $d_A^C(x)\le\Delta$,
where $A$ is the block of the model containing $x$.
Data sequences that are modelled using Kolmogorov complexity models
are stochastic; e.g., for some constant $c$:
\begin{itemize}
\item
  if a data sequence of length $N$ is $\Delta$-exchangeable
  (i.e., $\Delta$-random w.r.\ to the exchangeability model),
  it is $(\log N+c,\Delta+c)$-stochastic;
\item
  if a data sequence of length $N$ is $\Delta$-Markov
  (i.e., $\Delta$-random w.r.\ to the Markov model),
  it is $(2\log N+c,\Delta+c)$-stochastic;
\item
  if a data sequence of length $N$ is $\Delta$-Markov of second order,
  it is $(4\log N+c,\Delta+c)$-stochastic;
\item
  if a data sequence of length $N$ is $\Delta$-random
  w.r.\ to the IID model introduced in \cite{Vovk:1986-arXiv},
  it is $(\frac12\log N+c,\Delta+c)$-stochastic;
\item
  if a data sequence of length $N$
  is $\Delta$-exchangeable with one change point
  (i.e., $\Delta$-random w.r.\ to the changepoint model),
  it is $(2\log N+c,\Delta+c)$-stochastic.
\end{itemize}

\section{Quasi-universal e-variables}
\label{app:universal}

In this paper we are interested, at least implicitly,
in the universal e-test $\tilde E$ introduced in Lemma~\ref{lem:e-universal}.
It is a fundamental object in that
its components $\tilde E_A$ are the largest e-variables;
in this sense they are the most powerful e-variables.
By Proposition~\ref{prop:m-E},
$\tilde E_A$ is the likelihood ratio of the universal measure
to the null hypothesis $U_A$.
In the main part of the paper we discussed alternative hypotheses,
and the universal measure can be regarded as the universal alternative.

The way the universal measure is constructed
in the algorithmic theory of randomness
is by averaging over all subprobability measures
that are computable in a generalized sense
(see, e.g., \cite[Theorem~47]{Shen/etal:2017book},
the alternative proof).

The algorithmic theory of randomness, however,
provides only an ideal picture.
It can serve as a model for more practical approaches,
but it is not practical itself.
The two most conspicuous reasons
are that:
\begin{itemize}
\item
  the basic quantities used in the algorithmic theory of randomness,
  such as complexity or randomness deficiency,
  are not computable
  (they are only computable in a generalized sense,
  let alone efficiently computable);
  in particular, the universal alternative is not computable;
\item
  these basic quantities
  are only defined to within a constant
  (additive or multiplicative).
\end{itemize}
What we did in this paper can, however,
be regarded as a computable approximation to the ideal picture.
The idea (which is an old one) is to replace the universal alternative
by a Bayesian average of a statistical model
that is significantly richer than the null hypothesis.
In particular,
the UMM exchangeability e-variable discussed in the main part of this paper
can be regarded as a practical approximation to~$\tilde E$.

The justification that we had for the UMM e-variable is less convincing
than the justification for its ideal counterpart $\tilde E$:
it is the frequentist one given by Lemma~\ref{lem:frequentist}
and assuming that the observed data sequence
is generated by the UMM alternative.
Its advantage, however, is that this justification
does not involve an arbitrary constant factor.

It would be more in the spirit of the algorithmic theory of randomness
to use a different principle for choosing the alternative hypothesis:
instead of choosing an alternative probability measure likely to generate the data,
we could choose an alternative probability measure likely to lead
to a high likelihood ratio of the alternative to the null.

The general scheme of testing exemplified by this paper
is that we test a Kolmogorov compression model as null hypothesis,
and have a batch compression model with a more detailed summarising statistic as alternative.
This paper has the exchangeability model as the null
and a mixture of the first-order Markov model as the alternative.
We can imagine lots of other testing problems of this kind:
\begin{itemize}
\item
  The exchangeability model as the null,
  and the uniform mixture of the second-order Markov model as the alternative.
\item
  The exchangeability model as the null,
  and a mixture of the uniform mixtures of the $k$th order Markov models
  as the alternative;
  the weights $w_k$ for those should sum to 1, $\sum_k w_k = 1$,
  and tend to 0 as slowly as possible as $k\to\infty$.
\item
  The first-order Markov model as the null,
  and the second-order Markov model as the alternative.
\item
  The exchangeability model as the null and the changepoint model as alternative.
\item
  A changepoint at a postulated time $\tau$ as the null,
  and a changepoint at a different time as alternative.
  (In order to obtain confidence regions for the changepoint.)
\end{itemize}
We can call them instances of quasi-universal testing.

In information theory and statistics,
quasi-universal prediction and coding
(similar to quasi-universal testing discussed here) was promoted by Rissanen;
see, e.g., \cite{Rissanen:1983} and Gr\"unwald's review \cite{Grunwald:2007}.
Rissanen's suggestion for the weights $w_k$, $k=1,2,\dots$, that sum to 1
and tend to 0 slowly was
\[
  w_k
  :=
  \frac{1}{c k \log k \log\log k \log\log\log k\dots},
\]
where the denominator includes all terms that exceed 1
and $c\approx0.865$
is the normalizing constant \cite[Appendix~A]{Rissanen:1983}.
The word ``universal'', however, is sometimes used in a more limited sense
in information theory and statistics:
it may be universality, in some sense, for a given statistical model,
without attempting to make the statistical model wider.

Kolmogorov's ideal picture is based on computability,
but when discussing practical approximations
it may be useful to replace computability
by expressibility in a given language.
The idea of using expressibility in logic rather than computability
is much older than the algorithmic theory of randomness
(see, e.g., \cite[Sect.~1]{Martin-Lof:1969})
and goes back to Wald \cite{Wald:1937}.
This idea has led to higher-level algorithmic randomness,
as in \cite{Martin-Lof:1970} and, e.g., \cite{Kjos-Hanssen/etal:2010}.

In this paper we used the uniform prior on the Markov statistical model
to obtain our alternative hypothesis.
Another natural choice is Jeffreys priors \cite{Jeffreys:1961}.
However, in our current context they do not have any obvious advantages.
(Among their advantages in other contexts
are their invariance w.r.\ to smooth reparametrizations
and attaining minimax optimality in some cases \cite[Sect.~8.2]{Grunwald:2007}.)
They do not always exist
and many Bayesian statisticians find them objectionable
(see, e.g., \cite{Vovk/Shafer:2023}).
Using the uniform prior in this paper
leads to simple analytical expressions and efficient calculations.
Similar problems
(using the Markov model as alternative when testing exchangeability)
are considered in \cite{Ramdas/etal:2022} and \cite[Sect.~9.2.7]{Vovk/etal:2022book},
which use priors that are built on top of Jeffreys priors
but are not Jeffreys priors themselves.

The idea of quasi-universal testing is closely related
to Lindley's ``Cromwell's rule'' (see, e.g., \cite[Sect.~6.8]{Lindley:2006}).
A possible interpretation of Cromwell's rule in our context
is that, when designing a suitable e-variable,
we should think of all kinds of alternative models
(say, Markov models of all orders),
and then mix all of them.
Cromwell's rule as stated by Lindley is very general
and encompasses two aspects:
our statistical models should be as wide as possible,
and our priors should be diffuse (at least non-zero).

\section{Changepoint models}
\label{app:changepoint}

In this appendix we will discuss in greater detail
the changepoint compression models mentioned in the previous appendixes.
But first we discuss a changepoint alternative hypothesis
when testing exchangeability.

In the ideal picture,
we just use $\tilde E$ of Lemma~\ref{lem:e-universal} as e-test,
but in practice we could use
\begin{align}
  Q&(\{(z_1,\dots,z_N)\})
  :=
  \frac{1}{N-1}
  \sum_{n=1}^{N-1}
  \int_0^1
  \int_0^1
  \pi_0^{z_1+\dots+z_n}
  (1-\pi_0)^{n-z_1-\dots-z_n}\label{eq:line-1}\\
  &\qquad
  \pi_1^{z_{n+1}+\dots+z_N}
  (1-\pi_1)^{N-n-z_{n+1}-\dots-z_N}
  \d\pi_0
  \d\pi_1\label{eq:line-2}\\
  &=
  \frac{1}{N-1}
  \sum_{n=1}^{N-1}
  \Beta(z_1+\dots+z_n+1,n-z_1-\dots-z_n+1)\notag\\
  &\qquad
  \Beta(z_{n+1}+\dots+z_N+1,N-n-z_{n+1}-\dots-z_N+1)\notag\\
  &=
  \frac{1}{N-1}
  \sum_{n=1}^{N-1}
  \frac{(z_1+\dots+z_n)!(n-z_1-\dots-z_n)!(N-n+1)!}%
    {(z_{n+1}+\dots+z_N)!(N-n-z_{n+1}-\dots-z_N)!(n+1)!}
  \label{eq:CP}
\end{align}
as quasi-universal alternative probability measure.
The expression inside the double integral in \eqref{eq:line-1}--\eqref{eq:line-2}
is the likelihood of the observed data sequence
when the probability of 1
is $\pi_0$ before and including time $n\in\{1,\dots,N\}$
(the changepoint) and
is $\pi_1$ strictly after time $n$.
We average this likelihood over the uniform distribution for $(\pi_0,\pi_1)$
and then over the uniform distribution for the changepoint $n$.

The alternative Markov kernel corresponding to \eqref{eq:CP} is
\[
  Q_{N_1}(\{(z_1,\dots,z_N)\})
  =
  \frac{Q(\{(z_1,\dots,z_N)\})}%
    {\sum_{z'_1,\dots,z'_N:z'_1+\dots+z'_N=N_1}Q(\{(z'_1,\dots,z'_N)\})},
\]
where $N_1:=z_1+\dots+z_N$ is interpreted as the value of the summarising statistic.
Finally, we can compute the quasi-universal e-value as
\[
  E(z_1,\dots,z_N)
  :=
  \binom{N}{N_1}
  Q_{N_1}(\{(z_1,\dots,z_N)\}).
\]
We do not discuss efficient ways of computing this e-value
in this version of the paper.

\subsection*{Confidence regions}

Now suppose we believe that there is at most one changepoint
in a binary data sequence $z_1,\dots,z_N$
and would like to pinpoint its location.
To obtain a confidence region,
we need different null hypotheses.

The Kolmogorov compression model with the changepoint $\tau\in\{1,\dots,N-1\}$
has
\begin{equation}\label{eq:M_tau}
  t_\tau(z_1,\dots,z_N)
  :=
  \left(
    \sum_{n=1}^{\tau}
    z_n,
    \sum_{n={\tau+1}}^N
    z_n
  \right)
\end{equation}
as its summarising statistic.
Examples of probability measures that agree with this KCM are
\begin{multline*}
  P(\{(z_1,\dots,z_N)\})
  :={}\\
  \pi_0^{z_1+\dots+z_{\tau}}
  (1-\pi_0)^{\tau-z_1-\dots-z_{\tau}}
  \pi_1^{z_{\tau+1}+\dots+z_N}
  (1-\pi_1)^{N-\tau-z_{\tau+1}-\dots-z_N}
\end{multline*}
for $\pi_0,\pi_1\in[0,1]$.
Of course, these are not all probability measures that agree with \eqref{eq:M_tau};
those consist of all convex mixtures of the uniform probability measures
on $t_{\tau}^{-1}(k_0,k_1)$,
where $(k_0,k_1)\in\{0,\dots,\tau\}\times\{0,\dots,N-\tau\}$.

As alternative probability measure we can take \eqref{eq:CP} or,
which is slightly more natural,
its modification
\begin{multline*}
  Q_{\tau}(\{(z_1,\dots,z_N)\})
  :=
  \frac{1}{N-2}\\
  \sum_{n\in\{1,\dots,N-1\}\setminus\{\tau\}}
  \frac{(z_1+\dots+z_n)!(n-z_1-\dots-z_n)!(N-n+1)!}%
    {(z_{n+1}+\dots+z_N)!(N-n-z_{n+1}-\dots-z_N)!(n+1)!}
\end{multline*}
that only considers changepoint locations different from $\tau$,
the one we are testing.
The alternative Markov kernel becomes
\begin{multline*}
  Q_{\tau,K_0,K_1}(\{(z_1,\dots,z_N)\})
  ={}\\
  \frac{Q_{\tau}(\{(z_1,\dots,z_N)\})}%
    {\sum_{z'_1,\dots,z'_N:z'_1+\dots+z'_{\tau}=K_0,z'_{\tau+1}+\dots+z'_N=K_1}
    Q_{\tau}(\{(z'_1,\dots,z'_N)\})},
\end{multline*}
where $(K_0,K_1):=(z_1+\dots+z_{\tau},z_{\tau+1}+\dots+z_N)$
is the value of the summarising statistic.
Finally, we can compute the quasi-universal e-value as
\begin{equation}\label{eq:E_tau}
  E_{\tau}(z_1,\dots,z_N)
  :=
  \binom{\tau}{K_0}
  \binom{N-\tau}{K_1}
  Q_{\tau,K_0,K_1}(\{(z_1,\dots,z_N)\}).
\end{equation}

Once we have the e-values \eqref{eq:E_tau},
we have the e-confidence regions for the changepoint $\tau$:
at a significance level $\alpha$,
the e-confidence region is $\{\tau\mid E_{\tau}\le1/\alpha\}$
(see \cite{Vovk/Wang:2023}).
A natural direction of further research
is to find a computationally efficient version of the e-confidence regions
based on \eqref{eq:E_tau}.

\section{Neyman structure}
\label{app:Neyman}

In this appendix we assume, as usual in this paper,
that the sample space is finite.
(In this case every function on the sample space is bounded,
and we do not have to discuss completeness and bounded completeness separately;
in fact, the most relevant notion of completeness
for e-testing without this restriction
would have been ``semi-bounded completeness''
only involving functions that are bounded below.)

Let us say that a statistic (i.e., function on the sample space) $E$
is a \emph{similar} (or \emph{precise}) \emph{e-variable}
for a statistical model $\{P_{\theta}\mid\theta\in\Theta\}$
if $\int E \d P_{\theta} = 1$ for all $\theta\in\Theta$;
this is an analogue for e-testing
of Neyman and Pearson's \cite[Sects~IV(a) and~V(a)]{Neyman/Pearson:1933}
notion of a similar test.
And we say that a statistic $E$
has \emph{Neyman structure} w.r.\ to a sufficient statistic $T$
if $\E_{\theta}(E\mid T)=1$ $P_{\theta}$-a.s.\ for all $\theta\in\Theta$.
This is analogous to the standard notion of Neyman structure
(see, e.g., \cite[Sect.~4.3]{Lehmann/Romano:2022}).

A statistic $T$ is \emph{complete} if,
for any function $f$ on its range,
\[
  \Bigl(
    \text{$\E_{\theta}(f(T))=0$ for all $\theta\in\Theta$}
  \Bigr)
  \Longrightarrow
  \Bigl(
    \text{$f(T)=0$ $P_{\theta}$-a.s.\ for all $\theta\in\Theta$}
  \Bigr).
\]
The following is an analogue of Theorem 4.3.2 in \cite{Lehmann/Romano:2022}.

\begin{proposition}\label{prop:e-Neyman-2}
  Let $T$ be a sufficient statistic
  for a statistical model $\{P_{\theta}\mid\theta\in\Theta\}$.
  If $T$ is complete,
  a statistic is a similar e-variable if and only if it has Neyman structure w.r.\ to $T$.
  The condition that $T$ be complete is both sufficient and necessary.
\end{proposition}

\begin{proof}
  Suppose $T$ is complete.
  It is clear that a statistic that has Neyman structure is a similar e-variable.
  Now suppose $E$ is a similar e-variable.
  Set $f(T):=\E_{\theta}(E\mid T)$;
  $f$ can be chosen independent of $\theta$ since $T$ is sufficient.
  Since $\E_{\theta}(f(T)-1)=0$ for all $\theta$,
  $f(T)=1$ $P_{\theta}$-a.s.\ for all $\theta$,
  and so $E$ has Neyman structure.

  Now suppose that $T$ is not complete.
  Choose a $[-1,\infty)$-valued function $f$
  such that $\E_{\theta}(f(T))=0$ for all $\theta\in\Theta$
  but $f(T)\ne0$ with a positive $P_{\theta}$-probability for some $\theta\in\Theta$.
  Then $1+f(T)$ is a similar e-variable
  that does not have Neyman structure w.r.\ to $T$.
\end{proof}

For our purposes the following one-sided variation of having Neyman structure
is more useful (although it is much less widely applicable).
An \emph{e-variable} w.r.\ to a statistical model $\{P_{\theta}\mid\theta\in\Theta\}$
is a nonnegative random variable $E$
such that $\int E \d P_{\theta} \le 1$ for all $\theta\in\Theta$.
It has \emph{one-sided Neyman structure} w.r.\ to a sufficient statistic $T$
if $\E_{\theta}(E\mid T)\le1$ $P_{\theta}$-a.s.\ for all $\theta\in\Theta$.

Let us say that a statistic $T$ is \emph{supercomplete} if,
for any function $f$ on its range,
\begin{equation}\label{eq:supercomplete}
  \Bigl(
    \text{$\E_{\theta}(f(T))\le0$ for all $\theta\in\Theta$}
  \Bigr)
  \Longrightarrow
  \Bigl(
    \text{$f(T)\le0$ $P_{\theta}$-a.s.\ for all $\theta\in\Theta$}
  \Bigr).
\end{equation}
(It is clear that this property is stronger than completeness.)
Now we have the following analogue of Proposition~\ref{prop:e-Neyman-2}.

\begin{proposition}
  Let $T$ be a sufficient statistic
  for a statistical model $\{P_{\theta}\mid\theta\in\Theta\}$.
  If $T$ is supercomplete,
  a nonnegative random variable is an e-variable if and only if
  it has one-sided Neyman structure w.r.\ to $T$.
  The condition that $T$ be supercomplete is both sufficient and necessary.
\end{proposition}

\begin{proof}
  Suppose $T$ is supercomplete.
  It is clear that a nonnegative variable that has one-sided Neyman structure
  is an e-variable.
  Now suppose $E$ is an e-variable.
  Set $f(T):=\E_{\theta}(E\mid T)$.
  Since $\E_{\theta}(f(T)-1)\le0$ for all $\theta$,
  $f(T)\le1$ $P_{\theta}$-a.s.\ for all $\theta$,
  and so $E$ has one-sided Neyman structure.

  Now suppose that $T$ is not supercomplete.
  Choose a $[-1,\infty)$-valued function $f$
  such that $\E_{\theta}(f(T))\le0$ for all $\theta\in\Theta$
  but $f(T)>0$ with a positive $P_{\theta}$-probability for some $\theta\in\Theta$.
  Then $1+f(T)$ is an e-variable
  that does not have Neyman structure w.r.\ to $T$.
\end{proof}

The following two examples show that the notion of supercompleteness
is limited albeit not vacuous.

\begin{example}[exchangeability]
  The summarising statistic $t_E$ of the exchangeability compression model
  (we can set $t_E$ to the number of 1s in the data sequence)
  is supercomplete w.r.\ to the exchangeability statistical model
  (consisting of all exchangeable probability measures).
  This is because for each summary $k$ there exists
  an exchangeable probability measure concentrated on $t_E^{-1}(k)$.
  (And it clear that this argument is applicable to any batch compression model
  and the family of all probability measures that agree with it.)
\end{example}

\begin{example}[IID]
  On the other hand,
  $t_E$ is not supercomplete w.r.\ to the Bernoulli statistical model
  $(B_{\theta}\mid\theta\in(0,1))$
  (where $B_{\theta}$ is the probability measure on $\{0,1\}$
  satisfying $B_{\theta}(\{1\})=\theta$).
  The standard argument for completeness
  as given in \cite[Example 4.3.1]{Lehmann/Romano:2022}
  now fails.
  A function $f$ satisfying the first inequality in \eqref{eq:supercomplete}
  can be written as
  \begin{equation}\label{eq:counterexample}
    \sum_{k=0}^N
    f(k)
    \binom{N}{k}
    \rho^k
    \le
    0,
    \quad
    \text{for all $\rho\in(0,\infty)$},
  \end{equation}
  and under the supercompleteness we would have concluded that $f\le0$.
  But on the left-hand side of \eqref{eq:counterexample}
  we can have any polynomial of degree $N$,
  and a polynomial can be nonpositive without all its coefficients being nonpositive.
  An example is $-(\rho-1)^2$,
  which corresponds to the function
  \[
    f(k)
    :=
    \begin{cases}
      -1 & \text{if $k=0$}\\
      \frac{2}{N} & \text{if $k=1$}\\
      -\frac{2}{N(N-1)} & \text{if $k=2$}\\
      0 & \text{otherwise}.
    \end{cases}
  \]
\end{example}
\end{document}